%% file: main.tex
\documentclass[a4paper,final]{article}

\usepackage[margin=1.5in]{geometry}

\usepackage[mathscr]{euscript}

\usepackage{enumitem}
\usepackage{pgfplots,tabularx,tikz}
\usepgfplotslibrary{groupplots}
\pgfplotsset{compat=newest}
\usetikzlibrary{plotmarks,quotes,positioning}
\usetikzlibrary{arrows.meta}
\usepgfplotslibrary{patchplots}
\usepackage{siunitx}
\usepackage{arydshln,comment}
\usepackage{xcolor}
\definecolor{mycolor1}{rgb}{1.00000,0.00000,1.00000}%
\definecolor{mycolor2}{rgb}{0.00000,1.00000,1.00000}%
\definecolor{mycolor3}{rgb}{0.3,0.,0.9}%
\usepackage{listings}
\input{ex_shared}
\usepackage{mathtools,booktabs,multirow,amssymb}

\newtheorem{Corollary}[theorem]{Corollary}
\theoremstyle{definition}
\newtheorem{example}{Example}

\DeclarePairedDelimiter{\abs}{\lvert}{\rvert}
\DeclarePairedDelimiter{\norm}{\lVert}{\rVert}

\DeclareMathOperator{\adj}{adj}
\DeclareMathOperator{\offdiag}{offdiag}
\newcommand{\mmatrix}{\texttt{mmatrix}}
\newcommand{\matlab}{Matlab}
\newcommand{\lapack}{Lapack}

\ifpdf
\hypersetup{
	pdftitle={\texttt{mmatrix}: a \matlab{} toolbox for accurate computations on M-matrices using triplets},
	pdfauthor={B. Iannazzo, M. Najafi Kalyani and F. Poloni}
}
\fi

\begin{document}
	
\maketitle


\begin{abstract}
We introduce the \mmatrix{} toolbox, a \matlab{} software package for componentwise accurate computations with M-matrices described through left or right triplet representations. 
The core of the toolbox is a Fortran implementation, in the \lapack{} style, of the unblocked, recursive, and blocked versions of the GTH algorithm and its applications for the accurate computation of the solution of linear systems with M-matrix coefficient and nonnegative right-hand side, the LU factorization of an M-matrix and its inverse. 
These algorithms avoid subtractive cancellation; this property ensures high componentwise accuracy even for ill-conditioned problems. The toolbox contains also accurate algorithms for related problems, such as computing the Schur complement, the singular values, the square root of an M-matrix, and the solution of nonsymmetric algebraic Riccati equations associated with M-matrices.
The \matlab{} interface is based on an object-oriented implementation allowing one to use standard \matlab{} operations on M-matrices with triplet representation.
\end{abstract}

\textbf{Keywords } M-matrix, \mmatrix-toolbox, componentwise error

\textbf{MSC Codes} 65F30, 65Y15, 65F05, 15B48

\section{Introduction}
We introduce the \texttt{mmatrix}-toolbox, a \matlab{} package designed for
computations with M-matrices, available at \url{https://github.com/numpi/mmatrix}.

An M-matrix is a real matrix that can be written as $A = sI - B$, where $B$ is componentwise nonnegative, $I$ is the identity matrix and $s \ge \rho(B)$, where $\rho(B)$ is the spectral radius of $B$. These matrices appear in numerous applications ranging from applied probability to network models and discretization of PDEs.

Under mild assumptions an M-matrix possesses a \emph{triplet representation}, that is made of three objects: the off-diagonal elements of $A$, a positive vector $u$ and a nonnegative vector $v$ such that
\begin{subequations} \label{eqtriplet}
\begin{align} 
    \text{(right triplet)} \quad & Au = v, \label{eqrt} \\
    \text{(left triplet)} \quad  & u^T A = v^T. \label{eqlt}
\end{align}
\end{subequations}
Although this representation is not unique, it allows the design of componentwise accurate algorithms for some typical problems associated with M-matrices such as the solution of a linear system, the inversion, and so on, through the so-called GTH-like (or just GTH) algorithm \cite{Ocinneide,alfa2002accurate}, a subtraction-free variant of Gaussian elimination that exploits the M-matrix sign structure.

In contrast, the customary linear system solvers included in linear algebra libraries such as  \lapack{} do not guarantee componentwise accuracy. In particular, they typically deliver results within the relative norm accuracy predicted by the condition number, and this may lead to large relative errors in small entries.

The triplet representation allows one to get full componentwise accuracy in the linear system solvers, when the right-hand side is a nonnegative vector. From these linear system solvers it is possible to design componentwise accurate algorithms for more complicated problems such as the square root of an M-matrix and the solution of matrix equations associated with M-matrices.

The \mmatrix{} toolbox provides a practical framework to work with M-matrices with triplet representation in \matlab. It introduces a class \mmatrix{} that stores a matrix together with its triplet vectors and overloads standard operations, so that users can work with these objects in a
natural way. Warning messages alert the user when operations that may destroy componentwise accuracy are performed. The core computations are implemented in Fortran using variants of the GTH algorithm. These include unblocked routines, recursive versions, and blocked algorithms based on Level-3 BLAS. The structure follows the design of \lapack{} and algorithms are provided for both right and left triplets.
All routines are integrated into \matlab{} using object-oriented code, so that indexing and basic operations behave as expected.

The remainder of this paper is organized as follows. Section~\ref{sec2} surveys useful theoretical results on M-matrices, triplet representations, the GTH algorithm, including its componentwise accuracy properties, and the problems addressed by the Toolbox. Section~\ref{sec3} presents our Fortran implementations, covering the unblocked, recursive cache-oblivious, and blocked level-3 variants. Section~\ref{sec4} focuses on the \matlab{} side and introduces the \texttt{mmatrix} class. Section~\ref{sec5} reports numerical experiments confirming that our code returns componentwise accurate results in finite arithmetic, and Section~\ref{sec6} draws some conclusions.

\section{Preliminaries}\label{sec2}
In this section, we recall some properties of M-matrices and triplet representations, provide the necessary perturbation results, and discuss the main algorithms implemented in the Toolbox. We denote by $\boldsymbol 1$ and $\boldsymbol 0$ the ones and zero vectors, respectively. We write $A\ge 0$ [$A>0$] for a nonnegative [positive] matrix and analogously $v\ge 0$ [$v>0$] will be a nonnegative [positive] vector.

\subsection{Triplet representations}
We denote by $\offdiag(A)$ a matrix that has the same off-diagonal entries as $A\in\mathbb{R}^{n\times n}$, but its diagonal is replaced by zeros, i.e.,
\[
\offdiag(A)_{ij} = \begin{cases}
    A_{ij}, & i\neq j,\\
    0, & i=j.
\end{cases}
\]
We say that an M-matrix $A$ has a \emph{right triplet representation} $(\offdiag(A), u, v)$ if $u>0$, $v\geq 0$, and $Au=v$. Similarly, $A$ has a \emph{left triplet representation} $(\offdiag(A), u^T, v^T)$ if the last condition is replaced by $u^T A = v^T$.
\begin{example} \label{ex0}
    Let $\delta \geq 0$. The M-matrix $A_\delta=\begin{bsmallmatrix}
        2 & -2\\
        -1 & 1+\delta
    \end{bsmallmatrix}$ has a right triplet representation
    \[
    \left(\begin{bmatrix}
        0 & -2\\
        -1 & 0
    \end{bmatrix}, \begin{bmatrix}
        1\\1
    \end{bmatrix}, 
    \begin{bmatrix}
        0\\
        \delta
    \end{bmatrix}\right).
    \]
    There are different choices in the literature on how to store in a triplet the off-diagonal part of an M-matrix (with the plus sign or the minus sign, as a vector or as a matrix with ``holes'' in the diagonal); in any case, what matters is that the triplet specifies the off-diagonal elements of $A$.
\end{example}

A triplet representation determines uniquely the matrix $A$: indeed, from a right triplet, the diagonal of $A$ can be reconstructed as
\begin{equation}\label{eq:1}
A_{ii} = \frac{v_i-\sum_{j\ne i} A_{ij}u_j}{u_{i}},
\end{equation}
and an analogous formula holds for left triplets. An M-matrix may not possess a right or left triplet representation, or both. For instance, the matrices
\begin{equation}
    A_1 = \begin{bmatrix}
    1 & 0\\ -1 & 0
    \end{bmatrix},\qquad
    A_2 = \begin{bmatrix}
    1 & -1\\0 & 0
    \end{bmatrix},\qquad
    A_3 = \begin{bmatrix}
    1 & -1 & 0\\0 & 0 & 0\\
    -1 & 0 & 0
    \end{bmatrix},
\end{equation}
are M-matrices such that: $A_1$ has a left ($u=\boldsymbol 1$, $v=\boldsymbol 0$) but not a right triplet, $A_2$ has a right ($u=\boldsymbol 1$, $v=\boldsymbol 0$) but not a left triplet, $A_3$ has neither right nor left triplet.

M-matrices with a (right) triplets are called \emph{regular} in~\cite{Guo13regular,Guo16regularsingular}. Necessary and sufficient conditions for the existence of triplets can be obtained from \cite[Theorem 4]{bimm}, stated only for right triplets, but easily extended to left triplets. We recall the \emph{Frobenius normal form}: for every M-matrix $A$ there exists a permutation matrix $P$ such that $P^TAP$ is block upper triangular, i.e.,
\begin{equation} \label{frobenius_normal_form}
P^TAP = \begin{bmatrix}
    A_{11} &  A_{12} & \cdots & A_{1n}\\
    & A_{22} & \ddots & \vdots\\
    & & \ddots & \vdots\\
    & & & A_{\nu\nu},
    \end{bmatrix}    
\end{equation}
and each diagonal block $A_{ii}$ is irreducible.

\begin{theorem}\label{thm:1}
The matrix $A$ has a right [left] triplet if and only if in its Frobenius normal form~\eqref{frobenius_normal_form} for each singular diagonal block $A_{ii}$ (including the $1\times 1$ zero blocks), the blocks on the row [column] block other than $A_{ii}$ are zero, i.e., $A_{ij}=0$ [$A_{ji}=0$] for $j\ne i$.
\end{theorem}

From Theorem \ref{thm:1} we deduce a corollary on the existence of triplet representations under standard hypotheses.

\begin{Corollary}
If the M-matrix $A$ is nonsingular or singular and irreducible then it has a left and a right triplet representation.
\end{Corollary}

Theorem \ref{thm:1} includes also some further interesting examples of M-matrices with triplet representation, such as the ones with the block structure
\[
    A = \begin{bmatrix}
    A_{11} & 0\\0 & 0
    \end{bmatrix},
\]
where $A_{11}$ is a nonsingular (possibly empty) block and the rest of the matrix is $0$. These are the only M-matrices whose pseudoinverse is nonnegative \cite[Theorem 3.8]{1}.

Regarding uniqueness, it can be proved that if $A$ is irreducible singular then the triplet is unique (up to multiplication by a positive number) \cite[Thm.~3]{bimm} and the right triplet is such that $Au=\boldsymbol 0$. If $A$ is nonsingular, then there exist infinitely many triplets with independent vectors $u$: indeed, for each $v\geq 0$, $v\neq 0$ we can take $(\offdiag(A), A^{-1}v, v)$, which is a triplet under the mild assumption that $A^{-1}v>0$ (always verified when $A$ is irreducible).

\subsection{Perturbation results}

The natural perturbation model to use with triplets is that of componentwise error. The \emph{componentwise relative error} between a matrix $P\in\mathbb{R}^{m\times n}$ (typically nonnegative) and its perturbation $\tilde{P}$ is 
\begin{equation} \label{cdist}
    e(\tilde{P},P) = \max_{i,j} \frac{\abs{\tilde{P}_{ij}-P_{ij}}}{\abs{P_{ij}}}.
\end{equation}
In other words, to get a small componentwise relative error we require each entry of $\tilde{P}$ to be close (relatively) to the corresponding one of $P$, regardless of its magnitude. It makes sense to consider this distance if the entries of $P$ may have wildly different magnitudes, but each entry is computed individually: for instance probability transition formulas in a Markov chain. This error metric is stricter than the usual normwise relative error, because even tiny entries have to be accurate: for instance, we have
\[
P = \begin{bmatrix}
    0.9999\\
    0.0001
\end{bmatrix}, \quad \tilde{P} = \begin{bmatrix}
    0.9998\\0.0002
\end{bmatrix}, \quad e(\tilde{P}, P) = 1,
\]
an enormous componentwise error, even though the normwise relative error $\norm{\tilde{P} - P}_1 / \norm{P}_1 = 0.0001$ is small. The definition~\eqref{cdist} can be extended to vectors, as seen in the example. Note that the function $e(\cdot, \cdot)$ is not symmetric in its two arguments, like the usual relative error.

One usually sets $0/0=0$ and $b/0=\infty$ for $b\neq 0$ in~\eqref{cdist}, so that the zero entries in $P$ must be preserved for this error to be finite. With this extension, $e(\tilde{P},P) \leq \varepsilon$ is equivalent to
$(1-\varepsilon) P \leq \tilde{P} \leq P(1+\varepsilon)$, when $P$ is nonnegative.

We recall here the most notable componentwise perturbation results for M-matrices. 
\begin{theorem} \label{thm:componentwise perturbation bounds}
    Let $\tilde{A}, A\in\mathbb{R}^{n\times n}$ be M-matrices that have (left or right) triplets $(\offdiag(\tilde{A}), u, \tilde{v})$ and $(\offdiag(A),u,v)$ with the same vector $u$. Suppose that
    \[
    e(-\offdiag(\tilde{A}),-\offdiag(A)) \leq \varepsilon \quad \text{ and } \quad e(\tilde{v},v) \leq \varepsilon
    \]
    for some positive constant $\varepsilon<1$. Then,
    \begin{enumerate}
        \item $(1-\varepsilon)\operatorname{diag}(A) \leq  \operatorname{diag}(\tilde{A}) \leq (1+\varepsilon) \operatorname{diag}(A)$;
        \item $(1-\varepsilon)^n\operatorname{det}(A) \leq  \operatorname{det}(\tilde{A}) \leq (1+\varepsilon)^n \operatorname{det}(A)$;
        \item If $A$ is invertible, then so is $\tilde{A}$, and
        \[
        \frac{(1-\varepsilon)^{n-1}}{(1+\varepsilon)^{n}}A^{-1} \leq  \tilde{A}^{-1} \leq \frac{(1+\varepsilon)^{n-1}}{(1-\varepsilon)^{n}}A^{-1};
        \]
        \item If $\tilde{\lambda}$ and $\lambda$ are the smallest in modulus eigenvalues (so-called ``Perron eigenvalues'') of $\tilde{A}$ and $A$, respectively, then
        \[
        \frac{(1-\varepsilon)^{n}}{(1+\varepsilon)^{n-1}}\lambda \leq  \tilde{\lambda} \leq \frac{(1+\varepsilon)^{n}}{(1-\varepsilon)^{n-1}}\lambda;
        \]
        \item If $\tilde{A}$ and $A$ are singular irreducible, and $\tilde{z},z\in\mathbb{R}^{n}$  are the positive vectors in their kernels (so-called ``Perron eigenvectors''), normalized such that $\norm{\tilde{z}}_1=\norm{z}_1=1$, then
        \[
        \frac{(1-\varepsilon)^{n-1}}{(1+\varepsilon)^{n-1}}z \leq  \tilde{z} \leq \frac{(1+\varepsilon)^{n-1}}{(1-\varepsilon)^{n-1}}z.
        \]
    \end{enumerate}
\end{theorem}
These results appeared in~\cite{alfa,Ocinneide} for the $u=\boldsymbol{1}$ case, and (in part) in~\cite{xue2012accurate} for the general case. In the appendix we sketch a proof, and give a generalization where also $u$ is allowed to be perturbed. 

The bounds take a simpler form if we only care about first-order perturbations: for instance,
\[
\frac{(1-\varepsilon)^{n-1}}{(1+\varepsilon)^{n}}
= 1 - (2n-1)\varepsilon + \mathcal{O}(\varepsilon^2), \quad 
\frac{(1+\varepsilon)^{n-1}}{(1-\varepsilon)^{n}} = 1 + (2n-1)\varepsilon + \mathcal{O}(\varepsilon^2),
\]
hence $e(\tilde{A}^{-1}, A^{-1}) = (2n-1)\varepsilon + \mathcal{O}(\varepsilon^2)$.

Another appealing variant of this bound is a formulation which is symmetric in $\tilde{A}$ and $A$: if
\[
\frac{1}{1+\varepsilon}(-\offdiag(A)) \leq -\offdiag(\tilde{A}) \leq (1+\varepsilon) (-\offdiag(A)), \quad  \frac{1}{1+\varepsilon} v \leq \tilde{v} \leq (1+\varepsilon)v,
\]
then (with an analogous proof)
\[
\frac{1}{(1+\varepsilon)^{2n-1}} A^{-1} \leq \tilde{A}^{-1} \leq (1+\varepsilon)^{2n-1} A^{-1}.
\]


A consequence of Theorem~\ref{thm:componentwise perturbation bounds} is the surprising fact that the inverse of $A$ is always a well-conditioned function (in the componentwise sense) of the triplet that defines $A$, even if the matrix $A$ is ill-conditioned in the traditional normwise sense. 
%
Figure~\ref{fig:relations} displays the relations between the various quantities: while we can compute the inverse from a triplet in a well-conditioned way, we cannot compute reliably the inverse from the matrix entries, or a triplet from the matrix entries. Hence, when available, a triplet representation is the ideal way to represent an M-matrix. 
\begin{figure}
    \centering
    \begin{tikzpicture}
    \node[draw,rectangle] (triplet) {Triplet $(\offdiag(A),u,v)$};

    \node[draw,rectangle,below=1cm of triplet] (entries) {Matrix entries $A_{ij}$};

    \node[draw,rectangle,right=3.5cm of entries] (inverse) {Inverse $A^{-1}$};

    \draw[red, thick, -{latex[sep=1pt]}] (entries) -- node[below] {ill-conditioned} (inverse);

    \draw[blue, thick, -{latex[sep=1pt]}] ([xshift=-0.5em]triplet.south) -- node[left] {well-conditioned} ([xshift=-0.5em]entries.north);

    \draw[red, thick, -{latex[sep=1pt]}] ([xshift=0.5em]entries.north) -- node[right] {ill-conditioned} ([xshift=0.5em]triplet.south);

    \draw[blue, thick, -{latex[sep=2pt]}] (triplet) -- node[right=15pt] {well-conditioned} (inverse);
    
    \end{tikzpicture}
    \caption{Conditioning of transformations between various quantities, for an M-matrix $A$. The quantities $A^{-1}$ and $\operatorname{diag}(A)$ are well-conditioned functions of the triplet, even in the componentwise sense, but computing the triplet or the inverse from the matrix entries may be ill-conditioned.}
    \label{fig:relations}
\end{figure}
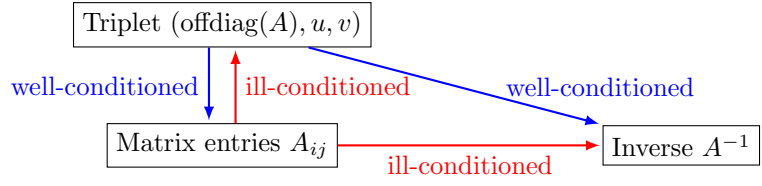
Fortunately, in many applications (such as Markov chains, network theory and differential equations) a triplet can be obtained analytically.

\begin{example}
    Consider $A = A_{\delta}$ as in Example~\ref{ex0}, with $\delta \ll 1$. 
    Take a perturbation of the form $\tilde{A} = A_{\tilde{\delta}}$, with its triplet $(\begin{bsmallmatrix}
        0 & -2\\
        -1 & 0
    \end{bsmallmatrix}, \begin{bsmallmatrix}
    1\\1
\end{bsmallmatrix}, \begin{bsmallmatrix}
    0\\
    \tilde{\delta}
\end{bsmallmatrix})$. If these two triplets are close in the componentwise sense, then we must have $e(\tilde{\delta},\delta)\leq \varepsilon$, and this implies that the inverses are close, since
\[
A_\delta^{-1} = \begin{bmatrix}
    \frac12 + \frac1{2\delta} & \frac1\delta\\
    \frac1{2\delta} & \frac1\delta\\
\end{bmatrix}.
\]
On the other hand, if we look at perturbations in terms of the matrix entries, then the assumption $e(\tilde{A},A) \leq \varepsilon$ (or also $\norm{\tilde{A}-A}/\norm{A}\leq \varepsilon$) does not guarantee that $e(\tilde{\delta},\delta)$ is small: for instance, changing $A_{\delta}$ to $A_{\delta/2}$ amounts to a very small relative perturbation of the entry $A_{2,2}$. So, under a small relative perturbation of the matrix entries, the two inverses $A_{\delta}^{-1}$ and $A_{\tilde{\delta}}^{-1}$ can be very far apart.
\end{example}

\subsection{Subtraction-free algorithms}
Since quantities such as $A^{-1}$ or $\operatorname{diag}(A)$ are well-conditioned functions of a triplet, the natural question is which algorithms can compute them in a stable way in the componentwise sense. The main computational tools are \emph{subtraction-free algorithms}, i.e., algorithms in which there are no subtractions of two real numbers with the same sign. For instance, the formula~\eqref{eq:1} to reconstruct $\operatorname{diag}(A)$ is subtraction-free, since the summands $v_i$ and $-A_{ij}u_j$ all have the same sign. In contrast, computing $v$ from $A$ and $u$ as $v = Au$ involves subtractions, since the summands $A_{ii}u_i$ and $A_{ij}v_j$ (for each $i\neq j$) have opposite signs.

Subtraction-free algorithms are automatically componentwise forward stable: in every product, division, or sum (of numbers with the same sign) the machine error on the result is within $\mathcal{O}(\mathsf{u})$, a small multiple of the machine precision, and errors in the operands are not amplified. A practical example of how to perform componentwise forward error analysis is in~\cite[Section 7]{NguP}; in particular Lemmas~7.1, 7.2, 7.8, 7.9. %

\subsection{The GTH algorithm}
The GTH algorithm, introduced in \cite{grassmann1985regenerative}, was later extended to triplet M-matrices by Alfa et al.~\cite{alfa}. Further developments and analysis can be found in \cite{Ocinneide, xxl2, xue2012accurate}.

Let $A$ be an M-matrix satisfying the right triplet relation \eqref{eqrt}.	In customary Gaussian elimination, the update at the step $\ell$
\begin{equation}\label{1}
A^{(\ell+1)}_{\ell+1:n,\ell:n} = A^{(\ell)}_{\ell+1:n,\ell:n}
- \frac{A^{(\ell)}_{\ell+1:n,\ell}}{A^{(\ell)}_{\ell,\ell}} A^{(\ell)}_{\ell,\ell:n},
\end{equation}
is not subtraction-free in computing the diagonal element $A_{ii}^{(\ell)}$; this may lead to large relative errors when cancellation occurs.
In contrast, the GTH algorithm maintains a sequence of triplets $(\offdiag(M^{(\ell)}),u_\ell,v_\ell)$ 
for the active matrices $M^{(\ell)}=A^{(\ell)}_{\ell:n,\ell:n}\in\mathbb{R}^{(n-\ell+1)\times(n-\ell+1)}$, where $M^{(1)}=A$, $u^{(1)}=u$, $v^{(1)}=v$, $\offdiag(M^{(\ell+1)})$ for $\ell \geq 1$ is obtained using \eqref{1},
\[
u^{(\ell+1)} = u_{\ell+1:n},\quad v^{(\ell+1)} = v^{(\ell)}_{\ell+1:n}-\frac{A^{(\ell)}_{\ell+1:n,\ell}}{A^{(\ell)}_{\ell,\ell}}v^{(\ell)}_\ell,
\]
and the diagonal element $A^{(\ell)}_{\ell,\ell}=M^{(\ell)}_{11}$ is accurately recomputed using formula~\eqref{eq:1} on the triplet of $M^{(\ell)}$. All computation is subtraction-free, as a result, the algorithm is always (forward) stable even without pivoting. 
An analogous algorithm can be designed if a left triplet is provided.

We describe the GTH algorithm in more detail in Section~\ref{sec3}, where we compare various blocked variants of it.

The GTH algorithm requires a larger number of arithmetic operations with respect to Gaussian elimination, the ones needed to construct $v^{(\ell)}$ and the ones in \eqref{eq:1}. These $2(n+1)^2$ operations however are negligible with respect to the cost of the elimination, $\frac{2}{3}n^3+\mathcal O(n^2)$ operations. 

\subsection{Problems solved with the GTH algorithm}

The GTH algorithm allows one to compute the LU factorization of a nonsingular M-matrix. This can be used to solve other basic problems involving M-matrices. 

If $A$ is a nonsingular M-matrix and $b\ge 0$, then, since $A^{-1}\ge 0$ \cite{bp}, 
the solution of the linear system $Ax=b$ is nonnegative. Computing the LU factorization $A=LU$, it can be solved as $U^{-1}(L^{-1}b)$ by forward/back substitution without cancellation, because the off-diagonal elements of $L$ and $U$ are nonpositive. Similarly, one can solve accurately multiple right hand side systems of the type $AX=B$ and $X^TA=B^T$ if $B\ge 0$. In this way one can get the accurate inverse of $A$ as the solution of $AX=I$.

A well-known property of M-matrices is that their principal submatrices and Schur complements are M-matrices as well \cite{bp}. If the matrix $A$ is singular and irreducible or nonsingular, then the diagonal minors (that are nonsingular \cite{bp}) and the Schur complements possess a right triplet that can be computed from the right triplet of $A$.

To show this fact choose a set of indices $\mathcal I\subset\{1,\ldots,n\}$ and denote by $\mathcal J$ its complement. The minor with indices in $\mathcal I$ will be denoted by $A_{\mathcal I\mathcal I}$ and its Schur complement is $S = A_{\mathcal J\mathcal J} - A_{\mathcal J\mathcal I}A_{\mathcal I\mathcal I}^{-1}A_{\mathcal I\mathcal J}$. We get the following right triplets
\[
    (A_{\mathcal I\mathcal I},u_{\mathcal I},v_{\mathcal I}-A_{\mathcal I\mathcal J}u_{\mathcal J}),\qquad
    (S,u_{\mathcal J},v_{\mathcal J}-A_{\mathcal J\mathcal I}A_{\mathcal I\mathcal I}^{-1}v_{\mathcal I}).
\]
Note that the computation of the triplet of the Schur complement requires the solution of a linear system with an M-matrix with right triplet coefficient.

An analogous argument can be used for the left triplets.

\subsection{Further problems}
The GTH algorithm and the triplet representation are not limited to solving linear systems or computing LU factorizations. They provide an effective tool for developing componentwise accurate algorithms for several other matrix problems that are challenging for standard approaches. In particular, when an M-matrix is given in triplet form, operations such as computing its square root, solving certain matrix equations, obtaining its singular values with high relative accuracy, or finding its smallest (in modulus) eigenvalue can be performed with limited or no subtractive cancellation. For each problem, we highlight the ideas that avoid cancellation and explain how these methods are realized in the \mmatrix-toolbox.

\subsubsection{Safe subtractions}
In some of the following algorithms, one must compute $P = I-\frac{1}{\zeta} A$, for a certain M-matrix $A$ and a scalar $\zeta>0$, to obtain a nonnegative matrix $P$ whose entries are used later in the algorithm. This requires subtractions (between two nonnegative numbers) to compute the diagonal entries of the result. Hence, this operation takes us out of the framework of subtraction-free algorithms. To make sure that the resulting algorithm is componentwise accurate, we must exploit a degree of freedom: the fact that $\zeta$ can be chosen arbitrarily, with the only restriction that $I-\frac{1}{\zeta} A \geq 0$. So instead of choosing the minimum possible value $\zeta_{\mathrm{opt}} = \max A_{ii}$, we select $\zeta = \gamma \zeta_{\mathrm{opt}}$, with $\gamma>1$ being an ``overshooting'' safety factor to ensure that the two operands of the subtraction $P_{ii} = 1 - \frac{A_{ii}}{\zeta}$ are not too close. In this way, if the computed result of the operation $\frac{1}{\zeta} A_{ii}$ is affected by a relative error $\varepsilon_0$ (due not only to the division itself, but also possibly to the accumulation of earlier errors in the computation of $A_{ii}$), then a quick error analysis shows that this error is propagated to a relative error of magnitude at most $\frac{1}{\gamma - 1}\varepsilon_0$ %

on $P_{ii}$. (We can neglect the machine error in the final subtraction $1 - \frac{A_{ii}}{\zeta}$, because it is a relative error on $P_{ii}$ of the order of the machine precision; for the values of $\gamma$ of interest it is negligible with respect to $\varepsilon_0$.) Hence by choosing, for instance, $\gamma = 1.5$, we can ensure that the relative error $\varepsilon_0$ is at most doubled. This technique was pioneered in this area in~\cite{xue2012accurate}, but related ideas are found in the choice of shifts in algorithms for highly-accurate bidiagonal SVD.

In the following algorithms, this safety factor $\gamma$ causes a mild reduction in the convergence speed (which is usually quadratic anyway), but it makes sure the algorithm is ``subtraction-safe'', even if not subtraction-free. 

\subsubsection{Accurate square root $A^{1/2}$}
If an M-matrix $A$ has a right triplet, then there exists the square root $A^{1/2}$ and it is a matrix with right triplet \cite{bimm}. Similarly, it can be proved that this holds for M-matrices with a left triplet as well. 

Algorithms for computing the right triplet of $A^{1/2}$ are presented in \cite{bimm}. According to the experiments, the best algorithm is the Newton method, written in the form of Cyclic Reduction \cite[Algorithm 2]{bimm}.

The Cyclic Reduction algorithm provides $A^{1/2}=\lim_\ell Z_\ell$, where $\{Z_\ell\}$ is computed with the iteration
\begin{equation}
    \left\{
    \begin{array}{l}
    W_0 = I - A,\qquad Z_0 = \frac{1}{2}(I + A)\\
    W_{\ell+1} = \frac{1}{4}W_{\ell}Z_{\ell}^{-1}W_{\ell},\qquad \ell = 0,1,2,\ldots \\
    Z_{\ell+1} = Z_{\ell}-\frac{1}{2} W_{\ell+1}.
    \end{array}
    \right.
\end{equation}

If $A$ is an M-matrix such that $I-A$ is nonnegative and $\rho(I-A)\le 1$ then $Z_\ell$ is an M-matrix and $W_\ell$ is nonnegative. If, moreover, $A$ has a triplet, then a triplet of $Z_0$ is obtained immediately and a triplet of $Z_{\ell+1}$ can be computed from a triplet of $Z_\ell$ by a componentwise accurate algorithm using the auxiliary sequence
\begin{equation} \label{eq:auxiliarysqrtm}
    \left\{
    \begin{array}{l}
    p_0 = v,\\ p_{\ell+1} = p_{\ell} + \frac{1}{2}W_{\ell+1} Z_\ell^{-1}p_\ell,
    \end{array}
    \right.
\end{equation}
and the triplet representation of $Z_{\ell+1}$ is ($\offdiag(Z_\ell-\frac{1}{2}W_{\ell+1}), u, p_{\ell+1}+\frac{1}{2}W_{\ell+1}u$).

The algorithm has possible cancellation in the initialization, because of the subtraction $I-A$.

It is possible to exploit the relation $A^{1/2} = \zeta^{1/2}(A/\zeta)^{1/2}$ to make $I-A/\zeta$ nonnegative with $\rho(I-A/\zeta)\le 1$, by choosing $\zeta =\gamma\max_{i=1,\ldots,n} A_{ii}$, with $\gamma>1$. This may also make the algorithm subtraction-safe for $\gamma$ sufficiently large because $|1-A_{ii}/\zeta|\geq 1-\frac{1}{\gamma}$, but on the other hand choosing a large $\gamma$ may slow down the convergence. The parameter $\gamma$ can be set by the user (the default choice is $\gamma=1.5$), together with the maximum number of iterations, and the tolerance $\varepsilon$ for the stopping criterion
\begin{equation}
    \max_{i,j}\frac{|(W_{\ell+1})_{ij}|}{(Z_\ell)_{ij}}\le \varepsilon.
\end{equation}

The corresponding algorithm for a matrix with left triplet is obtained by exploiting the fact $(A^T)^{1/2} = (A^{1/2})^T$.
\subsubsection{M-matrix algebraic Riccati equations} Another problem that has received considerable interest in the literature is that of \emph{M-matrix algebraic Riccati equations} (MARE) \cite{bimbook,doublingbook}. In a MARE, we want to compute the minimal nonnegative solutions $X,Y$ to the matrix equations
\begin{align} \label{mares}
    XDX - AX - XB + C = 0, \quad 
    YCY - BY - YA + D = 0
\end{align}
(or also only the first one; but in several applications both unknowns have a practical meaning). Here, $A\in\mathbb{R}^{m\times m}, B\in\mathbb{R}^{n\times n}, C, D^T, X, Y^T\in\mathbb{R}^{m\times n}$, and the four matrix coefficients can be arranged to form the M-matrix
\begin{equation} \label{mareM}
M = \begin{bmatrix}
    B & -D\\-C & A
\end{bmatrix}.    
\end{equation}
For dense equations, the state-of-the-art algorithm is the \emph{structured doubling algorithm} (SDA / ADDA, \cite[Chapter~6]{doublingbook}), i.e., the iteration
\begin{gather*}
    \begin{bmatrix}
        E_0 & Y_0\\
        X_0 & F_0
    \end{bmatrix} = \left(I + M\begin{bmatrix}
        \alpha I & 0\\
        0 & \beta I
    \end{bmatrix} \right)^{-1}\left(I - M\begin{bmatrix}
        \beta I & 0\\
        0 & \alpha I
    \end{bmatrix}\right),\\
    \begin{aligned}
    E_{k+1} &= E_k(I-Y_kX_k)^{-1}E_k, & X_{k+1} &= X_k + F_k(I-X_kY_k)^{-1}X_k E_k,\\
    F_{k+1} &= F_k(I-X_kY_k)^{-1}F_k, &
    Y_{k+1} &= Y_k + E_k(I-Y_kX_k)^{-1} Y_k F_k.
    \end{aligned}
\end{gather*}

When $\alpha, \beta$ are sufficiently small, one can show that $E_k, F_k, X_k, Y_k$ are nonnegative at each step, that $X_k \to X, Y_k \to Y$ (typically monotonically, and with quadratic convergence speed), and that all the matrices to invert are M-matrices, whose triplets can be computed analytically via auxiliary sequences in the same style as~\eqref{eq:auxiliarysqrtm}. We refer the reader to~\cite[Chapter~6]{doublingbook} for more details; in particular, the auxiliary sequences are~\cite[Equation (6.38)]{doublingbook}.


\subsubsection{Accurate singular values}
Standard algorithms for the singular value decomposition (SVD)  compute singular values of a matrix $A$ with error $|\hat{\sigma}_i - \sigma_i| = \mathcal{O}(\varepsilon)\|A\|$. Consequently, when $\sigma_i \ll \|A\|\varepsilon$, the computed $\hat{\sigma}_i$ may have no correct digits.  For certain M-matrices, however, it is possible to compute all singular values with high relative accuracy, meaning $|\hat{\sigma}_i - \sigma_i| = \mathcal{O}(\varepsilon)\sigma_i$, by exploiting a triplet representation.

Demmel and Koev~\cite{DemmelKoev2004} described this approach for $u = \mathbf{1}$. Here, we extend their framework by assuming that $A$ is an M-matrix with triplet representation $(\operatorname{offdiag}(A), u, v)$.

The algorithm consists of two parts.
The first part is to compute an $LDU$ factorization of $A$ using Gaussian elimination with diagonal pivoting, implemented via a GTH-like algorithm. That is, we run a variant of Algorithm~\ref{alg:dmxtf2} below, but at the beginning of each iteration of the \textbf{for} cycle we recover from the triplets all diagonal entries $U_{k,k},U_{k+1,k+1},\dots,U_{n,n}$, and permute rows and columns symmetrically so that the largest one becomes $U_{k,k}$. The algorithm then computes a factorization $PAP^T = LU$. If we set $D$ to be the diagonal matrix with entries $U_{1,1},U_{2,2},\dots,U_{n,n}$, and $\hat{U} = D^{-1}U$, then $PAP^T = LU = LD\hat{U}$ is a factorization in which $L$ and $\hat{U}$ are a lower and an upper triangular matrix, respectively, with unit diagonal. All entries in this factorization are componentwise accurate, thanks to the properties of the GTH algorithm.

Once this accurate factorization $PAP^T=LD\hat{U}$ is completed, we can write $A=P^TLD\hat{U}P$, then we can apply the accurate SVD algorithm from Demmel et al.~\cite[Algorithm 3.1]{Demmel1999}. Writing
\[
A=XDY^{T}, \qquad X=P^TL, \qquad Y=P\hat{U}^{T},
\]
the SVD is then computed via the following steps:
\begin{enumerate}
	\item Compute a QR factorization with pivoting $XD=QRP_{qr}$.
	\item Form the matrix $W=RP_{qr}Y^{T}$.
	\item Compute the SVD of $W$, $W=\bar{U}\Sigma V^{T}$, using a one-sided Jacobi method.
	\item Set $U=Q\bar{U}$.
\end{enumerate}
The overall computational cost is $\mathcal{O}(n^3)$ and does not depend on the condition number of $A$.

The error analysis in~\cite{Demmel1999} states that if $X$ and $Y$ are well-conditioned, then this algorithm delivers componentwise accurate singular values $\sigma_i$. (Note that the entries of the orthogonal matrices $U$ and $V$ are not guaranteed to be componentwise accurate, only the singular values.)

We wish to point out another subtle detail: if $u=\boldsymbol{1}$, as in~\cite{DemmelKoev2004}, then the diagonal pivoting described above coincides with complete pivoting, since in an M-matrix with $u=\boldsymbol{1}$ the largest-magnitude entry is always on the diagonal. If $u$ has entries of different magnitude, however, then the computed factorization is not Gaussian elimination with complete pivoting, and the condition numbers $\kappa(L)$ and $\kappa(\hat{U})$ could be larger, degrading the bounds from~\cite{Demmel1999}.

\subsubsection{Accurate computation of the smallest eigenvalue}
In many applications, such as Markov chains, the smallest eigenvalue of an M-matrix plays a crucial role. However,  algorithms like the QR method often produce inaccurate results when the smallest eigenvalue is either extremely small or ill-conditioned. Alfa, Xue, and Ye \cite{alfa2002accurate} proved that for M-matrices with triplet representation, the smallest eigenvalue is determined to high relative accuracy. 
Their algorithm is obtained by combining a shifted inverse iteration with the GTH algorithm. 
The idea is to avoid forming $A-\lambda I$ explicitly: instead, the triplet is updated at each step using nonnegative quantities. The  algorithm can be summarized as follows.

\begin{algorithm}[H]
	\caption{Shifted inverse iteration for smallest eigenvalue~(\lstinline{eigmin}) \protect{\cite[Algorithm~2]{alfa2002accurate}}}
	\label{alg:eigmin}
	\begin{algorithmic}[1]
		\REQUIRE $A \in \mathbb{R}^{n \times n}$ (M-matrix), $u,v \in \mathbb{R}^n$ forming a triplet, tolerance $\texttt{tol} > 0$, maximum number of iterations \lstinline{maxit}.
		\ENSURE The smallest eigenvalue $\lambda$ of $A$ and the corresponding eigenvector $u$
		\STATE $\lambda \gets \min(v \operatorname{\texttt{./}} u)$ \COMMENT{element-wise division}
		\STATE $v \gets v - \lambda u$
         \FOR{$k = 1, 2, \dots, \texttt{maxit}$}
        \STATE\COMMENT{Invariant: $(\operatorname{offdiag}(A), u, v)$ is a triplet for $A-\lambda I$}
        \IF{left triplet}
		\STATE Solve $z^T (A-\lambda I) = u^T$ for $z^T$ \COMMENT{using the triplet}
		\ELSE
		\STATE Solve $(A-\lambda I) z = u$ for $z$ \COMMENT{using the triplet}
		\ENDIF
		\STATE $p \gets {u} \operatorname{\texttt{./}} {z}$ 
		\STATE $\lambda \gets \lambda + \min(p)$
		\STATE $u \gets z / {\|z\|}_\infty$
		\STATE $v \gets u (p - \min(p))$
		\IF{$(\max(p) - \min(p)) / \lambda \leq \texttt{tol}$}
		\STATE \textbf{break}
		\ENDIF
		\ENDFOR
	\end{algorithmic}
\end{algorithm}

\section{Fortran implementation}\label{sec3}
The core numerical routines of the \texttt{mmatrix}-toolbox are implemented in Fortran for the sake of efficiency and compiled as Mex files for use in \matlab. The toolbox provides several algorithms for the LU factorization of an M-matrix, including unblocked, recursive, and blocked implementations, described in the following subsections. Following \lapack{} conventions, we use the second and third letter to specify the matrix storage format: \lstinline{mm} for an M-matrix given via a right triplet, and \lstinline{ml} for an M-matrix given via a left triplet.

\subsection{Unblocked algorithms: \texttt{dmmtf2} and \texttt{dmltf2}}
The routines \texttt{dmmtf2} and \texttt{dmltf2} contain the unblocked (BLAS level-2)
algorithms for the accurate computation of the LU of an M-matrix using the right and left triplet representations, respectively. They play the same role as \texttt{dgetf2} in \lapack, providing the simplest but slowest algorithm for LU factorization. We present pseudocode for both in Algorithm~\ref{alg:dmxtf2}.
\begin{algorithm}
\caption{Generalized GTH Algorithm, unblocked version}
\label{alg:dmxtf2}
\begin{algorithmic}[1]
\REQUIRE $A \in \mathbb{R}^{n \times n}$ (M-matrix), $u,v \in \mathbb{R}^n$ forming a triplet.
\ENSURE $L,U$ which form an LU factorization of $A$
\STATE $L \gets I_n$, $U\gets A$
\FOR{$k = 1$ to $n-1$}
    \STATE \COMMENT{Invariants: $LU=A$, $(U_{k:n,k:n}, u_{k:n}, v_{k:n})$ is a triplet}
    \STATE{Recover the pivot $U_{k,k}$ from the triplet}
    \STATE $L_{k+1:n,k} \gets U_{k+1:n,k} / U_{k,k}$
    \STATE $U_{k+1:n,k} \gets 0$
    \STATE $U_{k+1:n,k+1:n} \gets U_{k+1:n,k+1:n} - L_{k+1:n,k}U_{k,k+1:n}$
    \STATE Update $v_{k+1:n}$ to maintain the triplet invariant
\ENDFOR
\end{algorithmic}
\end{algorithm}
The algorithm updates at each step a partial factorization
\begin{equation} \label{block_partial_factorization}
    A = \begin{bmatrix}
        A_{11} & A_{12}\\
        A_{21} & A_{22}
    \end{bmatrix} = \begin{bmatrix}
        L_{11} & 0\\
        L_{21} & I_{n-k}
    \end{bmatrix}
    \begin{bmatrix}
        U_{11} & U_{12}\\
        0 & S
    \end{bmatrix}, \quad u = \begin{bmatrix}
        u_1\\ u_2
    \end{bmatrix},
    v = \begin{bmatrix}
        v_1 \\ v_2
    \end{bmatrix},
\end{equation}
with $S = A_{22} - A_{21}A_{11}^{-1}A_{12}$, expanding at each step the dimension of the  first factored block from $k-1$ to $k$. After each step of factorization, the vector $v$ is updated so that $(U_{k+1:n,k+1:n}, u_{k+1:n}, v_{k+1:n})$ remains a triplet.

The left and right triplet variants differ only in two lines. The first one is the way the pivot is recovered:
\[
U_{kk} =
\frac{
    v_k - \sum_{i=k+1}^{n} u_i U_{i k}
}{u_k}, \quad
U_{kk} =
\frac{
    v_k - \sum_{j=k+1}^{n} U_{k j} u_j
}{u_k}
\]
for the left and right version respectively; this update formula follows directly from the first column or row of the invariant triplet relation; note that it is cancellation-free, since $U_{ij} \leq 0$ for $i\neq j$. The second difference is the way $v$ is updated to maintain the triplet invariant,
\[
v_{k+1:n} = v_{k+1:n} - \frac{v_k}{U_{kk}}U_{k,k+1:n}, \quad v_{k+1:n} = v_{k+1:n} - L_{k+1:n,k} v_k.
\]

In the Fortran code, using the classical \lapack{} style, we do not allocate additional memory for $L$ and $U$, but rather overwrite the matrix $A$ in place at each step: the nontrivial entries of $U$ are stored in the upper triangular part of $A$, and those of $L$ in the strictly lower triangular part of $A$.

\subsection{Recursive GTH kernels: \texttt{dmmtrf2} and \texttt{dmltrf2}}
The routines \texttt{dmmtrf2} and \texttt{dmltrf2} implement a recursive,
cache-oblivious variant of the GTH factorization. Their structure follows
the same idea as \texttt{dgetrf2} in \lapack: the matrix is partitioned recursively, and each block is reduced using a combination of triangular
solves and updates. 

This version of the algorithm operates on the more general case of a rectangular matrix $A \in \mathbb{R}^{m\times n}$, with $m\geq n$, in which $A_{ij} \leq 0$ for each $i\neq j$. Also in this case the diagonal entries of $A$ can be computed accurately by the (rectangular) left triplet relation $u^T A = v^T$, with $u \in\mathbb{R}^m,v\in\mathbb{R}^n$, or a right triplet relation $Au=v$, with $u \in\mathbb{R}^n,v\in\mathbb{R}^m$; in the latter case, only the first $n$ entries of $v$ need to be provided, since the matrix $A$ only has $n$ diagonal entries to recover.

We set
\begin{equation} \label{block_factorization}
    A = \begin{bmatrix}
        A_{11} & A_{12}\\
        A_{21} & A_{22}
    \end{bmatrix} = \begin{bmatrix}
        L_{11} & 0\\
        L_{21} & L_{22}
    \end{bmatrix}
    \begin{bmatrix}
        U_{11} & U_{12}\\
        0 & U_{22}
    \end{bmatrix}, \quad u = \begin{bmatrix}
        u_1\\ u_2
    \end{bmatrix},
    v = \begin{bmatrix}
        v_1 \\ v_2
    \end{bmatrix},
\end{equation}
where the leading block has size $\nu = \lfloor n/2\rfloor$. 
The algorithm call itself recursively twice, once on the left block $\begin{bsmallmatrix}
    A_{11}\\
    A_{21}
\end{bsmallmatrix}\in\mathbb{R}^{m\times \nu}$ and once on the Schur complement $S$. The reason for this peculiar asymmetric partitioning is that $\begin{bsmallmatrix}
    A_{11}\\
    A_{21}
\end{bsmallmatrix}$ is formed by consecutive entries using the Fortran column-major memory layout; hence at some point in the recursion the algorithm operates on a consecutive block of memory that fits in the cache, in the style of cache-oblivious algorithms~\cite{frigo1999cache}. We present pseudocode as Algorithm~\ref{alg:dmxtrf2}.
\begin{algorithm}
\caption{Generalized GTH Algorithm, recursive version}
\label{alg:dmxtrf2}
\begin{algorithmic}[1]
\REQUIRE $A \in \mathbb{R}^{m \times n}$ (M-matrix), $u\in \mathbb{R}^m$ (if left) or $\mathbb{R}^n$ (if right) , $v_1\in\mathbb{R}^n$.
\ENSURE $L,U$ which form an LU factorization of $A$
\IF{$n = 1$}
\STATE Recover $A_{11}$ from the triplet relation
\RETURN $L = \frac{1}{A_{11}}A,\, U = A_{11}$
\ELSE
\STATE $\nu \gets \left\lfloor n/2 \right\rfloor$ \COMMENT{Size of the blocks with index 1}
\STATE Compute a triplet for $ \begin{bsmallmatrix} A_{11}\\ A_{21} \end{bsmallmatrix} $
\STATE Call the function recursively on $ \begin{bsmallmatrix} A_{11}\\ A_{21} \end{bsmallmatrix} $, obtaining $\begin{bsmallmatrix}
    L_{11}\\L_{21}
\end{bsmallmatrix},\,U_{11}$
\STATE $U_{12} \gets L_{11}^{-1}A_{12}$
\STATE $S \gets A_{22} - L_{21}U_{12}$
\STATE Update $v_2$ so that $(S,u_2,v_2)$ is a triplet
\STATE Call the function recursively on $S$, obtaining $L_{22},U_{22}$
\RETURN $\begin{bmatrix}
    L_{11} & 0\\
    L_{21} & L_{22}
\end{bmatrix}, \, 
\begin{bmatrix}
    U_{11} & U_{12}\\
    0 & U_{22}
\end{bmatrix}$
\ENDIF
\end{algorithmic}
\end{algorithm}

This time, in the implementation there is a more substantial difference between the left and right case. In the left case, a left triplet for the left panel is readily obtained without computation as $(\begin{bsmallmatrix}
    A_{11}\\
    A_{21}
\end{bsmallmatrix}, u, v_1)$. In the right case, a triplet is $(\begin{bsmallmatrix}
    A_{11}\\A_{21}
\end{bsmallmatrix}, u_1, v_1 - A_{12}u_2)$, and we need additional storage to store the result of the computation $v_1 - A_{12}u_2 \in \mathbb{R}^\nu$. (We cannot overwrite the entries of $v$, since they are all needed in the following parts of the algorithm.)

In the \lapack{} style, this storage space must be given as a work buffer provided by the caller. The right variant of this recursive version \lstinline{dmmtrf2} needs a work buffer of size $n$: $\nu$ entries to store the entries of $v_1 - A_{12}u_2$, and additional $n-\nu$ entries to pass to the two recursive calls as \emph{their} work buffers.

The base case where $A$ is a single vector is trivial and requires little discussion, but to completely specify Algorithm~\ref{alg:dmxtrf2} we describe how to update $v_2$. In the right version, from the block factorization~\eqref{block_partial_factorization} and $Au=v$ we get
\[
\begin{bmatrix}
    U_{11} & U_{12} \\
    0 & S
\end{bmatrix}
\begin{bmatrix}
u_1 \\
    u_2
\end{bmatrix}= 
\begin{bmatrix}
    L_{11} & 0\\
    L_{21} & I
\end{bmatrix}^{-1}
\begin{bmatrix}
    v_1\\
    v_2
\end{bmatrix} = 
\begin{bmatrix}
    L_{11}^{-1}v_1 \\
    v_2 - L_{21}L_{11}^{-1}v_1
\end{bmatrix}.
\]
From the second block row we can read off $Su_2 = v_2 - L_{21}L_{11}^{-1}v_1$, which is the required triplet. In the left version, from the block factorization
\[
A = \begin{bmatrix}
    L_{11} & 0\\
    L_{21} & S
\end{bmatrix}\begin{bmatrix}
    U_{11} & U_{12}\\
    0 & I
\end{bmatrix}
\]
and $u^T A = v^T$ we get
\[
\begin{bmatrix}
    u_1^T & u_2^T
\end{bmatrix}
\begin{bmatrix}
    L_{11} & 0\\
    L_{21} & S
\end{bmatrix}
= \begin{bmatrix}
    v_1^T & v_2^T
\end{bmatrix}
\begin{bmatrix}
    U_{11} & U_{12}\\
    0 & I
\end{bmatrix}^{-1}
= \begin{bmatrix}
    v_1^T U_{11}^{-1} & v_2^T - v_1^T U_{11}^{-1}U_{12}
\end{bmatrix},
\]
and from the second block we read off the triplet $u_2^T S = v_2^T - v_1^T U_{11}^{-1}U_{12}$.

\subsection{Blocked algorithms: \texttt{dmmtrf} and \texttt{dmltrf}}
The routines \texttt{dmmtrf} and \texttt{dmltrf} implement the triplet-preserving
elimination in a blocked, panel-wise fashion using the unblocked functions
\texttt{dmmtf2} and \texttt{dmltf2}.

These algorithms take as input a block size $n_b$; the only significative case is $n_b < n$, otherwise we fall back to one of the previous variants. We update at each step a partial factorization of the form~\eqref{block_partial_factorization}. We describe the first step; then in the successive ones we apply the same strategy to the Schur complement $S$.

\begin{algorithm}
\caption{Generalized GTH Algorithm, blocked version}
\label{alg:dmxtf22}
\begin{algorithmic}[1]
\REQUIRE $A \in \mathbb{R}^{n \times n}$ (M-matrix), $u,v \in \mathbb{R}^n$ forming a triplet, block size $n_b$
\ENSURE $L,U$ which form an LU factorization of $A$
\STATE $L \gets I_n$, $U\gets A$
\FOR{$k = 1, n_b+1, 2n_b+1, \dots$ up to $n-1$}
    \STATE \COMMENT{Invariants: $LU=A$, $(U_{k:n,k:n}, u_{k:n}, v_{k:n})$ is a triplet}
    \STATE{$\ell\gets \min(k+n_b-1, n)$} \COMMENT{last index of the working block}
    \STATE{Compute $\hat{v}_{k:\ell}$ such that $U_{k:\ell,k:\ell}, u_{k:\ell}, \hat{v}_{k:\ell}$ is a triplet}
    \STATE{Call the unblocked algorithm to factor $U_{k:\ell,k:\ell}$, obtaining $L_{k:\ell,k:\ell}$ and $U_{k:\ell,k:\ell}$}
    \STATE $L_{\ell+1:n,k:\ell} \gets U_{\ell+1:n,k:\ell} \, U_{k:\ell,k:\ell}^{-1}$
    \STATE $U_{\ell+1:n,k:\ell} \gets 0$
    \STATE $U_{k:\ell,\ell+1:n} \gets L_{k:\ell,k:\ell}^{-1} \, U_{k:\ell,\ell+1:n}$
    \STATE $U_{\ell+1:n,\ell+1:n} \gets U_{\ell+1:n,\ell+1:n} - L_{\ell+1:n,k:\ell}U_{k:\ell,\ell+1:n}$
    \STATE Update $v_{\ell+1:n}$ to maintain the triplet invariant
\ENDFOR
\end{algorithmic}
\end{algorithm}

At step $1$, the diagonal block $A_{11} \in \mathbb{R}^{n_b \times n_b}$ is factored using the unblocked algorithm; for this, we need the triplet representation $(\offdiag(A_{11}), u_1,\hat v_1)$, where $u_1$ is made of the first $n_b$ elements of $u$ and
\[
\hat{v}_1 = v_1 - A_{12}u_2 \quad \text{(right)}, \qquad
\hat{v}_1^T = v_1^T - u_2^T A_{21} \quad \text{(left)}.
\]
Work buffer space is needed for the $n_b$ entries of $\hat{v}_1$. After computing the factorization
$A_{11} = L_{11}U_{11},$
the remaining blocks are obtained via triangular solves,
\[
U_{12} = L_{11}^{-1}A_{12}, \qquad
L_{21} = A_{21}U_{11}^{-1},
\]
and the trailing matrix is updated by 
$ S = A_{22} - L_{21}U_{12}$, with triplet $(\offdiag(S), u_2, v_2)$, where the triplet vector $v_2$ is updated as follows
\begin{equation}\label{eqrl}
	v_2 :=  v_2 - L_{21}L_{11}^{-1}v_1 \quad \text{(right)}, \qquad
	v_2^T := v_2^T - v_1^T U_{11}^{-1}U_{12} \quad \text{(left)}.
\end{equation}
\subsection{Recursive blocked algorithms: \texttt{dmmtrff} and \texttt{dmltrff}}
The routines \texttt{dmmtrff} and \texttt{dmltrff} combine panel blocking with
recursive blocked algorithms. At each iteration, the diagonal block $A_{11}$ is factorized using the recursive kernel (\texttt{dmmtrf2} or \texttt{dmltrf2}) applied to the corresponding  triplet. The subsequent steps are identical to those in the blocked algorithms:
\[
U_{12} = L_{11}^{-1}A_{12}, \qquad
L_{21} = A_{21}U_{11}^{-1}, \qquad
S = A_{22} - L_{21}U_{12}.
\]
The triplet vector is updated in the same way as in the blocked case.
\subsection{Matrix inversion via triangular factors (\texttt{dmmtri}) }\label{subri}	
The \texttt{dmmtri} algorithm, analogous to \lapack's \texttt{dgetri}, computes the inverse of a matrix from its LU factors. Let $A \in \mathbb{R}^{n \times n}$ and
$A = LU,$ where $L$ is unit lower triangular and $U$ is upper triangular. The inverse is given by
$A^{-1} = U^{-1} L^{-1}.$ The algorithm proceeds in two steps. First, it computes $U^{-1}$, and then it solves $X L = U^{-1}$ for $X=A^{-1}$. More precisely, let
\begin{equation*}
	L =
	\begin{bmatrix}
		L_{11} & 0 \\
		L_{21} & L_{22}
	\end{bmatrix}, \quad
	X =
	\begin{bmatrix}
		X_1 & X_2
	\end{bmatrix}, \quad
	U^{-1} =
	\begin{bmatrix}
		U_1 & U_2
	\end{bmatrix},
\end{equation*}
where $L_{11}$ and $L_{22}$ are unit lower triangular matrices.
Substituting into $XL = U^{-1}$ gives
\begin{equation*}
	\begin{bmatrix}X_1 & X_2\end{bmatrix}
	\begin{bmatrix}
		L_{11} & 0 \\
		L_{21} & L_{22}
	\end{bmatrix}
	=
	\begin{bmatrix}X_1 L_{11} + X_2 L_{21}& X_2 L_{22}\end{bmatrix}
	=
	\begin{bmatrix}U_1 & U_2\end{bmatrix}.
\end{equation*}
This yields $X_2 L_{22} = U_2$ and $ X_1 L_{11} = U_1 - X_2 L_{21}$.
Hence,
\begin{align*}
	X_2  = U_2 L_{22}^{-1}, \quad	X_1 = \left(U_1 - X_2 L_{21}\right) L_{11}^{-1}.
\end{align*}
The computation proceeds blockwise from right to left. Each block column $X_2$ is obtained by solving a triangular system with $L_{22}$. The algorithm overwrites the input matrix $A$ with its inverse.
\subsection{Fortran routine interface}

The GTH-based routines follow a \lapack{}-style design. The prefix
\texttt{D} denotes double precision, while the next characters specify the triplet type: \texttt{MM} for right, and \texttt{ML} for left.

The unblocked and blocked algorithms are implemented as standard Fortran
subroutines with interface
\begin{lstlisting}[language=Fortran]
    SUBROUTINE DMxYYY(M, N, A, LDA, U, V, WORK, NB, INFO)
\end{lstlisting}
where \texttt{x} selects the triplet representation and \texttt{YYY} is one
of \texttt{TF2}, \texttt{TRF}, or \texttt{TRFF}.

The arguments follow \lapack{} conventions: \texttt{M} and
\texttt{N} define the matrix dimensions; \texttt{A(LDA,*)} stores the
matrix in column-major format and is overwritten by the $LU$ factors;
\texttt{LDA} is the leading dimension; \texttt{U} and \texttt{V} define the
triplet representation; \texttt{WORK} provides workspace (used only in
blocked variants); \texttt{NB} is the block size; and \texttt{INFO} is the
exit flag.

The recursive kernels are defined as follows.
\begin{lstlisting}[language=Fortran]
    RECURSIVE SUBROUTINE DMxTRF2(M, N, A, LDA, U, V, WORK, INFO)
\end{lstlisting}
These routines implement cache-oblivious factorizations. Their arguments
coincide with those above, except that the block size parameter is omitted.
The workspace array \texttt{WORK} is required only for the right-triplet
variant (\texttt{dmmtrf2}), where it stores temporary values of
\texttt{V}.

\subsection{Work buffer requirements}
	The different algorithms have distinct memory requirements, summarized in Table~\ref{tabw}.
	\begin{table}[htbp]
		\centering
		\begin{tabular}{lccc}
			\toprule
			Routine & Work buffer size & BLAS level & Algorithm structure \\
			\midrule
			\texttt{dmmtf2} & 0 & 2 & unblocked \\
			\texttt{dmltf2} & 0 & 2 & unblocked \\
			\midrule
			\texttt{dmmtrf2} & $n$ & 2 and 3 & recursive \\
			\texttt{dmltrf2} & 0 & 2 and 3 & recursive \\
			\midrule
			\texttt{dmmtrf} & $n_b$ & 3 & blocked (unblocked panel) \\
			\texttt{dmltrf} & $n_b$ & 3 & blocked (unblocked panel) \\
			\midrule
			\texttt{dmmtrff} & $2n_b$ & 3 & blocked (recursive panel) \\
			\texttt{dmltrff} & $n_b$ & 3 & blocked (recursive panel) \\
			\bottomrule
		\end{tabular}
		\caption{Work buffer requirements for each variant. The buffer size is given in terms of the block size $n_b$ or matrix dimension $\min(m,n)$. BLAS level 2/3 indicates a mix of level-2 and level-3 operations.}
		\label{tabw}
	\end{table}
	The work buffer is used to store temporary vector values that would otherwise be overwritten during the factorization. In the blocked variants, the buffer holds the values of $v$  for the current panel before they are updated; this allows algorithms to preserve the triplet relationship throughout the factorization.
	
\section{The \texttt{mmatrix} class}\label{sec4}
The \texttt{mmatrix}-Toolbox is implemented in Matlab throught the \texttt{mmatrix} class that provides an object-oriented interface for working with M-matrices with triplet representation. The main properties of an \texttt{mmatrix} object are:
	\begin{itemize}
		\item \texttt{obj.M}: a copy of the M-matrix (with a recomputed accurate diagonal),
		\item \texttt{obj.u}, \texttt{obj.v}: the triplet vectors $u > 0$ and $v \ge 0$, satisfying either \eqref{eqrt} or \eqref{eqlt}.
		\item \texttt{obj.left}: logical flag indicating the triplet type,
		\item \texttt{obj.n}: matrix dimension.
	\end{itemize}
    An object is created as follows:
	\begin{lstlisting}
		obj = mmatrix(M, u, v) % default: left=false
		obj = mmatrix(M, u, v, left=true)
	\end{lstlisting}
    During construction, diagonal entries are recomputed from the triplet relation to avoid subtractive cancellation; the diagonal of the constructor parameter \lstinline{M} is ignored.

If \texttt{v} is not provided, it is computed as $Mu$ or $u^TM$ depending on
the triplet type. This computation of $v$ is componentwise accurate only when $M$ is diagonal; otherwise, it is carried out using standard arithmetic, with a warning indicating that componentwise accuracy may be lost.

\subsection{Main functions}
The \texttt{mmatrix} class overloads standard \matlab{} operations to operate
directly on M-matrices with triplet representation while preserving componentwise accuracy. The main functions are summarized in Table~\ref{tab2}.
	\begin{table}[htbp]
	\begin{center}
		\begin{tabularx}{\textwidth}{>{\ttfamily}l X}
			\hline
			Function & Description \\
			\hline
			\lstinline!double!(A) & Return the full matrix $A$ in double precision. \\
			\lstinline!vpa!(A,digits) & Return a variable-precision representation of $A$. \\
			\lstinline!lu!(A) & Compute an LU factorization with \texttt{dmmtrff} or \texttt{dmltrff}.\\
			\lstinline!mldivide!(A,B) & Solve $A \backslash B$ accurately when $B \ge 0$. \\
			\lstinline!mrdivide!(B,A) & Solve $B / A$ accurately when $B \ge 0$. \\
			\lstinline!inv!(A) & Compute $A^{-1}$ with reduced cancellation error. \\
            \lstinline!svd!(A) & Compute the SVD of $A$.
            \\
			\lstinline!det!(A) & Compute the determinant of $A$. \\
			\lstinline!sqrtm!(A) & Compute a componentwise accurate matrix square root. \\
			\lstinline!sda!(M,n1,opts) & Solve the MARE~\eqref{mares} using the structured doubling algorithm, where $M$ as in~\eqref{mareM} is an M-matrix, and \texttt{n1} denotes the dimension of the $(1,1)$-block $B$.  \\
			\lstinline!schur!(A,p,q) & Compute Schur complement (as an \mmatrix{}).
            \\
            \lstinline!eigmin!(A,tol,maxit) & Compute the smallest eigenvalue of $A$\\
			\hline
		\end{tabularx}
	\end{center}
			\caption{List of implemented functions in \mmatrix-toolbox}
	\label{tab2}
	\end{table}
	These methods enable standard linear algebra workflows while preserving the
	structure and numerical properties of triplet M-matrices.
	\subsubsection*{Overloading Operators}
	Standard \matlab{} operators are overloaded to support \texttt{mmatrix}
	objects:
    \begin{itemize}
 \item[\textbf{(\texttt{()})}] The syntax \lstinline{A(I,J)} is supported to return a submatrix. Here \lstinline{I,J} can be numbers, vectors of indices (e.g., \lstinline{1:5}), logical vectors (e.g., \lstinline{A.u<1}), expressions using the \matlab{} keyword \lstinline{end} (e.g., \lstinline{1:end-1}). If the matrix is a principal submatrix, it is returned as an \lstinline{mmatrix} object, otherwise as a \lstinline{double}.
 \item[\textbf{(\texttt{+})}]
If $A=(M_1,u,v_1)$ and $B=(M_2,u,v_2)$ are two \lstinline{mmatrix} objects with the same triplet vector $u$ and the same direction (left or right), then their sum is the \lstinline{mmatrix} object $A + B = (M_1 + M_2,\; u,\; v_1 + v_2)$. The operator also supports the addition of a nonnegative diagonal matrix. In all other cases, the sum is computed using standard arithmetic, with a warning that componentwise accuracy may be lost.
\item[\textbf{(\texttt{-})}]
The difference of two M-matrices is generally not an M-matrix. For this reason, subtraction is not implemented within the triplet representation. The operands are converted to double precision and the subtraction is computed as $\text{\lstinline!double!}(A)-\text{\lstinline!double!}(B)$.
\item[\textbf{(\texttt{*}}, \textbf{\texttt{.*})}]
Multiplication by a positive scalar $\alpha>0$ preserves the triplet representation and is performed through the update $(M,u,v)\mapsto (\alpha M,u,\alpha v)$.

The class also supports multiplication by a diagonal matrix $D=\operatorname{diag}(a)$ with $a_i>0$.

For a right triplet $Au=v$, left scaling gives $(DA)u=Dv$, and is represented by $(DM,u,Dv)$.

Similarly, right scaling gives $(AD)(D^{-1}u)=v$, hence the updated triplet becomes $(MD,D^{-1}u,v)$. Analogous formulas are used for left triplets.

If positivity is not preserved, the operation falls back to standard floating-point arithmetic.
\item[\textbf{(\texttt{.'}}, \textbf{\texttt{'})}]
The transpose is implemented by replacing $(M,u,v)$ with $(M^{T},u^{T},v^{T})$.
\item[\textbf{(\texttt{\^{}})}]
The only matrix power implemented in triplet form is the square root $A^{1/2}$, computed using \lstinline!sqrtm!\texttt{(.)}. For all other exponents, the matrix is converted to double precision and \matlab{}'s built-in matrix power routine is used.
\item[\texttt{(diag(.)})]
The diagonal entries are reconstructed from the triplet relation and returned by the \lstinline!diag! function.
\item[(\texttt{norm(.)})]
Norms are evaluated using \matlab{}'s \texttt{norm}: $\|A\|_p=\|\text{\lstinline!double!}(A)\|_p$. 
\end{itemize}
\subsection{The \texttt{gettriplet} function}
The function \texttt{gettriplet} is used to generate test M-matrices in triplet form as \lstinline{mmatrix} objects for numerical experiments, similar to \matlab{}'s \lstinline{gallery}. It uses \matlab{} name-value inputs such as \texttt{size} (default $n=10$) and \texttt{left} (to choose between right or left triplets).

Two main test cases are available:
\begin{itemize}
	\item \lstinline!'rand'!: Generates random M-matrix with random off-diagonal entries and random triplet vectors $u$ and $v$.
	\item \lstinline!'fem'!: Returns a tridiagonal M-matrix arising from a finite element discretization with step size $h$ (default $10^{-2}$).  The output \texttt{out} also contains the right-hand side vector $b$.
\end{itemize}

An example of LU factorization for a random case is:
\begin{lstlisting}
	A = gettriplet('rand', size=n, left=1); % right triplet: left=0
	[L, U] = lu(A); % Computed with componentwise accuracy
\end{lstlisting}

Similarly, for the FEM test case, the associated linear system $Ax = b$ can be solved as:
\begin{lstlisting}
	[A, out] = gettriplet('fem', size=n, h=1e-3);
	x = A \ out.b; 
\end{lstlisting}		
	\section{Numerical experiments}\label{sec5}
This section evaluates the accuracy and performance of the \mmatrix-toolbox on several test problems. 
All experiments were performed using \matlab{} R2025b on a
computer equipped with a 64-bit Intel Core i7-13620H processor. To assess numerical accuracy, we use both componentwise and normwise errors. 
Let $x, x^\ast \in \mathbb{R}^n$ and $X, X^\ast \in \mathbb{R}^{m \times n}$, where $x^\ast$ and $X^\ast$ denote either the exact solution or a high-precision solution computed via \texttt{vpa}. We define the componentwise and normwise errors by
	\[
	e_c := e(x, x^\ast) = \max_i \frac{|x_i - x_i^\ast|}{|x_i^\ast|}, 
	\qquad 
	e_n := \frac{{\|x - x^\ast\|}_2}{{\|x^\ast\|}_2},
	\]
	and
	\[
	E_c := e(X, X^\ast)=\max_{i,j} \frac{|X_{ij} - X^*_{ij}|}{|X^*_{ij}|}, 
	\qquad 
	E_n := \frac{{\|X - X^\ast\|}_F}{{\|X^\ast\|}_F}.
	\]
We compare the proposed GTH-based algorithms against standard \matlab{} functions (\texttt{lu}, \texttt{inv}, \texttt{eig}, \texttt{svd}) in terms of both componentwise accuracy and execution time.
	\begin{example}\label{ex1}
We consider a linear system of the form $A x = b$,
arising from a one-dimensional finite element discretization.  The matrix $A \in \mathbb{R}^{n \times n}$ is given by
\[
A =
\begin{bmatrix}
    1 + h\sigma & -1 &        &        &        \\
    -1          & 2  & -1     &        &        \\
    & \ddots & \ddots & \ddots &   \\
    &        & -1     & 2      & -1 \\
    &        &        & -1     & 1 + h\sigma
\end{bmatrix},
\]
where $h = \frac{1}{n+1}$ is the grid spacing.
This matrix has a triplet representation with
$u = \mathbf{1}$ and $v = (h\sigma, 0, \dots, 0, h\sigma)^T$. Since $A^T=A$, we have also $u^TA=v^T$. The right-hand side vector $b$ is of the form $
b_i = h \sin(\pi h i)$, $i = 1, \dots, n$, and $b_1=b_1/2, b_n=b_n/2$.
The matrix size $n$ and the parameter $\sigma$ are indicated in the figures. 

Figure \ref{fig1} shows that all GTH-based implementations preserve near machine-precision accuracy in both the normwise and componentwise errors, while the standard LU factorization exhibits large componentwise inaccuracies as the problem becomes more ill-conditioned. 

Table \ref{tab3} reports the average CPU times for the different Fortran implementations together with the pure \matlab{} GTH implementation and the standard LU factorization. The blocked and recursive variants significantly improve computational time compared with the unblocked implementation, especially for large matrix dimensions.
		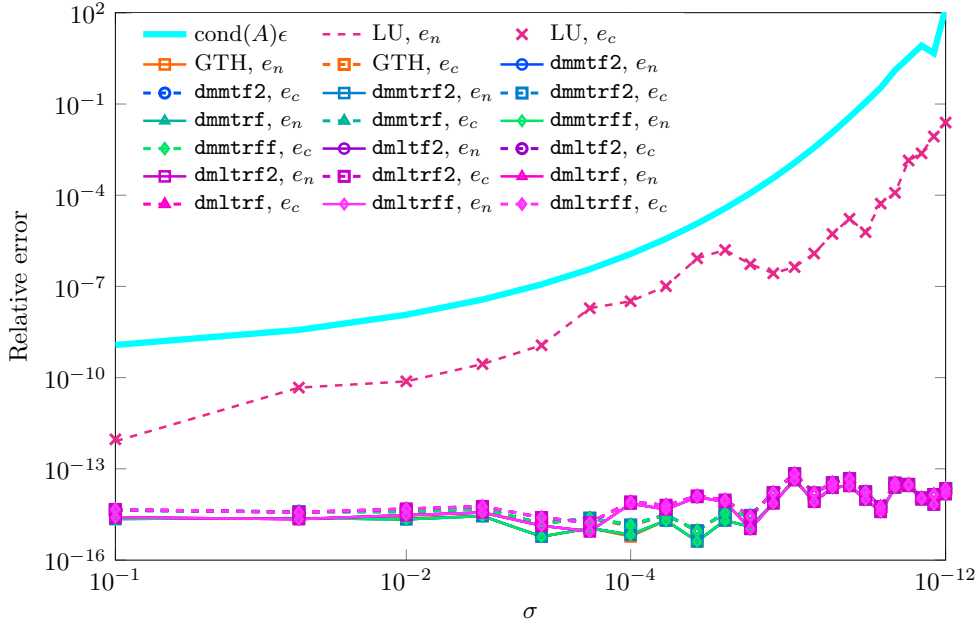
\begin{figure}[!htbp]
			\centering
			\input{fig1}
			\caption{Example \ref{ex1}: Accuracy comparisons for $n=512$. GTH denotes the pure \matlab{} implementation of the GTH algorithm with the right triplet representation, whereas LU refers to \matlab’s built-in LU factorization routine based on \texttt{dgetrf}.}
            \label{fig1}
		\end{figure}
	\begin{table}[htbp]
		\centering
		\caption{Example \ref{ex1}: Average CPU time (seconds) over 10 runs for both the computation of the LU factors by each method and the solution of the finite element linear system with $\sigma = \num{e-09}$. The symbol ``--'' indicates that the method failed to compute the solution within 300 seconds.}
		\resizebox{\textwidth}{!}{ 
			\begin{tabular}{l c c c c c c}
				\toprule
				{Methods}\quad\textbackslash\,\,\, Dimension $n$: &1000&4000& 5000&8000&10000&15000\\
				\midrule
				\matlab's \lstinline{lu}& 0.0079 & 0.2189 & 0.3951 & 1.3394 & 2.5363 &8.0799\\
				\midrule
				\multicolumn{7}{l}{\textit{Right‑triplet algorithms:}} \\
				GTH (pure \matlab{}) & 1.1245 & 90.7768 & 178.2144 & -- & -- &--\\
				\lstinline{dmmtf2} (unblocked) & 0.0249 & 7.5489& 16.5645 & 76.8099 & 154.4136& --\\
				\lstinline{dmmtrf2} (recursive)& 0.0074 & 0.1672 & 0.3028 & 1.1287 & 2.1804&7.0821 \\
				\lstinline{dmmtrf} (blocked)& 0.0069 & 0.1894 & 0.3341 & 1.1841 & 2.2801 &7.3640\\
				\lstinline{dmmtrff} (blocked recursive) & 0.0077 & 0.1691 & 0.3067 & 1.1801 & 2.2202&7.2710 \\
				\midrule
				\multicolumn{7}{l}{\textit{Left‑triplet algorithms:}} \\
				GTH (pure \matlab{})& 1.1330 & 90.9462 & 178.2964 & -- & -- &--\\
				\lstinline{dmltf2} (unblocked)& 0.0225 & 7.5124 & 16.4713 & 76.4553 & 153.7987& --\\
				\lstinline{dmltrf2} (recursive)& 0.0063 & 0.1859 & 0.3148 & 1.1925 & 2.2564&7.3533 \\
				\lstinline{dmltrf} (blocked)& 0.0068 & 0.1966 & 0.3370 & 1.2187 & 2.2434&7.3828 \\
				\lstinline{dmltrff} (blocked recursive)& 0.0065 & 0.1724 & 0.3147 & 1.1975 & 2.2044 &7.2814\\
				\bottomrule
		\end{tabular}}  
		\label{tab3}
	\end{table}
		\end{example}
\begin{example}\cite[Example 1]{alfa2002accurate}\label{ex3} 
This example evaluates the accuracy and efficiency of \lstinline{eigmin} compared to \matlab’s \lstinline{eig}, for an ill-conditioned M-matrix. 
Let $\hat{\lambda}$ and $\tilde{\lambda}$ denote the approximations obtained by \lstinline{eigmin} and \lstinline{min(eig(A))}, respectively. We consider the $n \times n$ matrix $A = I - P$, where 
\[
P =
\begin{pmatrix}
    0 & 1 & 0 & \cdots & 0 \\
    0 & 0 & 1 & \cdots & 0 \\
    \vdots & \vdots & \vdots & \ddots & \vdots \\
    0 & 0 & 0 & \cdots & 1 \\
    \delta & 0 & 0 & \cdots & 0
\end{pmatrix}.
\]
The triplet representation of $A$ is defined by the vectors $u = \mathbf{1}$ and $v = (0, 0, \dots, 0, 1-\delta)^T$, such that $Au = v$. The exact smallest eigenvalue is given by $\lambda = 1 - \delta^{1/n}$. As $\delta \to 0$, the eigenvalue $\lambda$ approaches $1$, and the matrix becomes ill-conditioned.
	
Table~\ref{tab_3} shows the results for $n=100$ as $\delta$ varies. The \lstinline{eigmin} function maintains high componentwise accuracy even when $\delta = 10^{-30}$. In contrast, \matlab’s \lstinline{eig} fails significantly once $\delta < 10^{-18}$, yielding relative errors $e(\lambda, \tilde{\lambda}) > 1$. This failure arises from subtractive cancellation: forming $A = I - P$ in floating-point arithmetic leads to the loss of the tiny difference $1 - \delta^{1/n}$ that defines the eigenvalue.

Table~\ref{tab4} fixes $\delta = 10^{-9}$ and varies the dimension $n$. While \lstinline{eigmin} remains robust with errors near machine precision, the accuracy of \lstinline{eig} degrades by several orders of magnitude. Furthermore, \lstinline{eigmin} is faster than \lstinline{eig} for $n \ge 500$.
	\begin{table}[htbp]
		\centering
		\caption{
        Example~\ref{ex3}: Accuracy comparison for $n=100$ with varying $\delta$ ($\texttt{tol}=\num{e-16}$, $\texttt{maxit}=1000$).
        }
		\begin{tabular}{c c c c}
			\toprule
			$\delta$ & $\lambda$ & $e(\lambda,\hat{\lambda})$ &  $e(\lambda,\tilde{\lambda})$  \\
			\midrule
			\num{e-03} &\num{ 6.67e-02} &\num{ 1.46e-15} &\num{ 7.65e-14} \\
			\num{e-06} &\num{ 1.29e-01} &\num{ 2.15e-16} &\num{ 5.47e-12} \\
	    	\num{e-12} &\num{ 2.41e-01} &\num{ 3.45e-16} &\num{ 7.74e-10} \\
			\num{e-18} &\num{ 3.39e-01} &\num{ 7.69e-15} &\num{ 1.95e+00} \\
			\num{e-24} &\num{ 4.25e-01} &\num{ 1.70e-15} &\num{ 1.36e+00} \\
			\num{e-30} &\num{ 4.99e-01} &\num{ 1.11e-16} &\num{ 1.00e+00} \\
			\bottomrule
		\end{tabular}
        \label{tab_3}
	\end{table}
\begin{table}[htbp]
	\centering
	\caption{Example \ref{ex3}: Accuracy and CPU-time comparison for $\delta = \num{e-09}$  with varying $n$ ($\texttt{tol}=\num{e-16}$, $\texttt{maxit}=1000$).}
	\begin{tabular}{c c c c c c c}
			\toprule
		$n$ & $\lambda$ & $e(\lambda,\hat{\lambda})$ & {CPU-time} && $e(\lambda,\tilde{\lambda})$ & {CPU-time} \\
		\midrule
		\num{10}   & \num{8.74e-01} & \num{2.54e-16} & \num{0.0002} && \num{1.18e-09} & \num{0.0001} \\
		\num{50}   & \num{3.39e-01} & \num{1.31e-15} & \num{0.0008} && \num{1.34e-08} & \num{0.0004} \\
		\num{100}  & \num{1.87e-01} & \num{7.41e-16} & \num{0.0014} && \num{9.17e-09} & \num{0.0039} \\
		\num{500}  & \num{4.06e-02} & \num{1.88e-15} & \num{0.0226} && \num{1.05e-08} & \num{0.1539} \\
		\num{1000} & \num{2.05e-02} & \num{4.23e-15} & \num{0.1148} && \num{8.21e-09} & \num{0.5234} \\
		\num{5000} & \num{4.14e-03} & \num{2.03e-14} & \num{4.9531} && \num{1.97e-07} & \num{30.1934} \\
		\bottomrule
	\end{tabular}
    \label{tab4}
\end{table}	
\end{example}
\begin{example}\label{ex2}
We compare three approaches for inverting the matrix from Example~\ref{ex1}: the standard \matlab\ \lstinline{inv} function, the \lstinline{inv} method from the \mmatrix-toolbox, and a Fortran implementation \lstinline{inv}* that computes LU factors with \lstinline{dmmtrff} and then solves 
$AX=I$ for the inverse.
Here, the exact inverse of the matrix $A$ was computed using \lstinline{vpa}.

Figure~\ref{fig3} illustrates normwise and componentwise errors versus $\sigma$ (decreasing from left to right). For \matlab's \lstinline{inv}, both $E_n$ and $E_c$ grow as $\sigma$ decreases, reaching about $10^{-3}$ for the smallest $\sigma$ which follows the condition number bound $\operatorname{cond}(A)\epsilon$. In contrast, the GTH‑based \lstinline{inv} and \lstinline{inv}* keep both $E_n$ and $E_c$ close to machine precision.

Table~\ref{tab_time} shows that \matlab's \lstinline{inv} is slightly faster than the other two algorithms. However, the GTH-based \lstinline{inv} and \lstinline{inv}* achieve accuracy near the machine precision.
\begin{figure}[htbp]
		\centering
	 	\input{fig3}
		\caption{Example \ref{ex2}: Results for matrix inversion with $n=128$.}
        \label{fig3}
\end{figure}
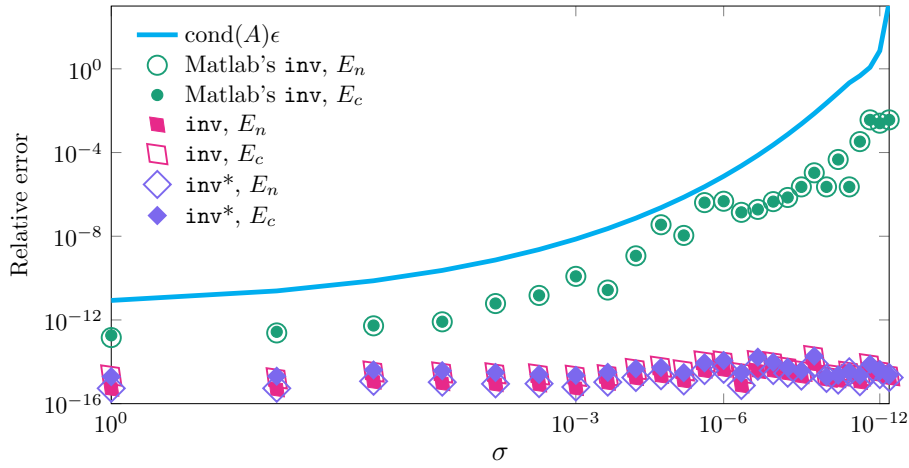
\begin{table}[htbp]
    \centering
     \caption{Example \ref{ex2}: Comparison of CPU times (s) for three matrix inversion methods applied to triplet M-matrices of different dimensions and $\sigma=10^{-9}$.}	
     \label{tab_time}
    \begin{tabular}{c c c c}        
        \toprule
         $n$ & \protect{{\matlab{}'s}} \texttt{inv} & \texttt{inv} & \texttt{inv}* \\
         &&&\\
        \midrule
1000&0.0319&0.0231&0.0249\\
3000&0.2682&0.2828&0.2782\\
6000&1.4266&1.7516&1.9767\\
9000&4.1958&5.7934&6.5056\\
12000&9.9581&13.8217&15.5551\\
15000&26.8983&29.0633&31.2356\\
        \bottomrule
    \end{tabular}
\end{table}
\end{example}
\begin{example}\label{ex4}
This example evaluates the accuracy of the singular value decomposition (SVD) computed by the \mmatrix-toolbox for ill-conditioned M-matrices. Following the construction in \cite{DemmelKoev2004}, we generate random M-matrices from their off-diagonal entries and the vectors $u$ and $v$ as follows.

For a given dimension $n \in \{20, 30\}$, the off-diagonal entries $a_{ij}$ ($i \neq j$) are drawn uniformly from $[-1,0]$. The vector $v$ is constructed by setting each component $v_i = r \cdot 10^k$, where $r \in [0,1]$ is uniformly random and $k \in [-40,-20]$ is an integer. Then, for each row $i$, both the off-diagonal entries $a_{ij}$ ($j \neq i$) and $v_i$ are multiplied by a factor $r \cdot 10^k$, where $r \in [0,1]$ and $k \in [-100,100]$ is an integer.

We consider two cases for the vector $u$:
\begin{itemize}
    \item a well-conditioned case: \lstinline{u = rand(n,1) .* exp(randn(n,1))};
    \item an ill-conditioned case: \lstinline{u = rand(n,1) .* exp(15 * randn(n,1))}.
\end{itemize}

Figure~\ref{fig4} shows the computed singular values for both dimensions $n = 20$ and $n = 30$. The results
obtained by SVD from \mmatrix-toolbox are close to the high-precision SVD (\lstinline{vpa}).
In contrast, the standard \matlab{} \lstinline{svd} loses accuracy for singular values below the  line 
$\sigma_1 \epsilon$.
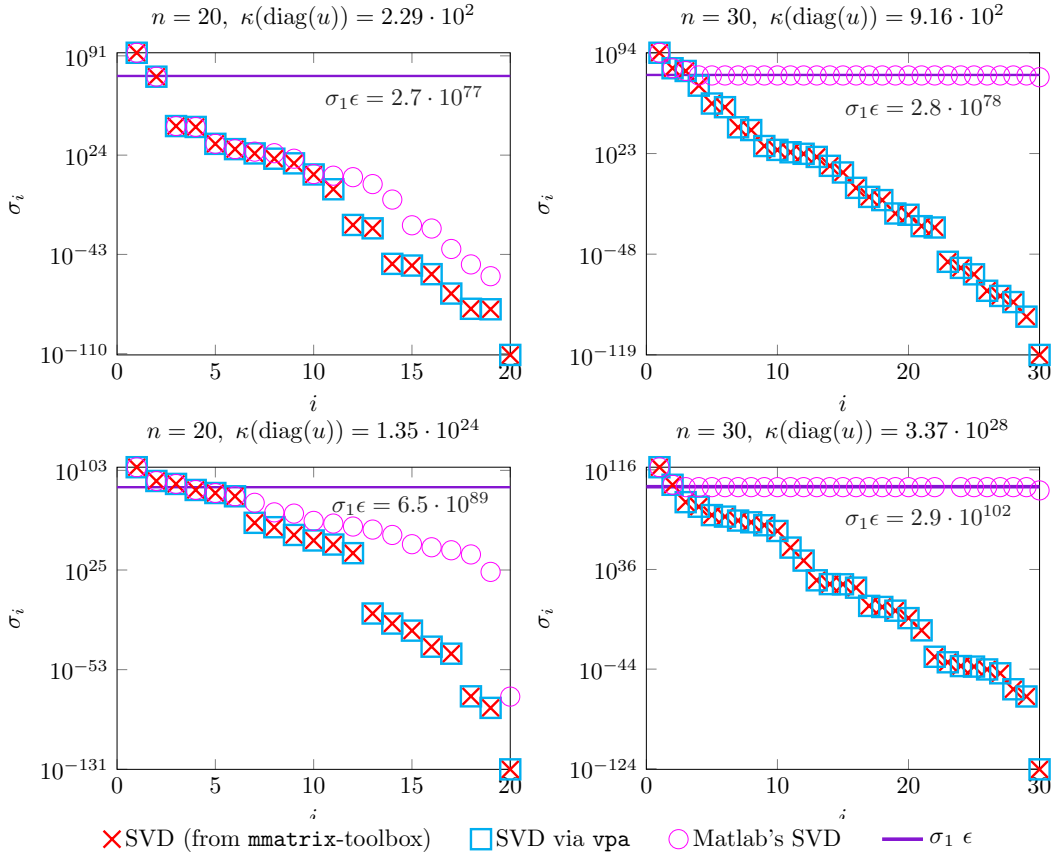
\begin{figure}[htbp]
		\centerline{\input{figsvd}}
		\caption{Example \ref{ex4}: Singular values for random M-matrices of size $n = 20$(left) and $n = 30$(right).}
		\label{fig4}
	\end{figure}
\end{example}
\begin{example}\label{ex_svd2}
This example provides a second test for the SVD routines using the M-matrix and its triplets from Example~\ref{ex1}.

Table~\ref{tab_svd2} shows that the standard \matlab{} SVD loses componentwise accuracy
as the condition number increases. 
In contrast, the SVD algorithm implemented in the \mmatrix-toolbox maintains componentwise
errors close to machine precision. Moreover, since the normwise error is dominated by the largest singular values, $e_n$ stays near machine precision for both methods and hides the errors in the smaller singular values.
\begin{table}[htbp]
	\centering
	\caption{Example \ref{ex_svd2}: Comparison of componentwise and normwise errors for the vector of all singular values produced by \matlab's SVD and the \mmatrix-toolbox.}
	\begin{tabular}{ccccccc}
	\toprule
	& & & \multicolumn{2}{c}{\matlab's SVD} & \multicolumn{2}{c}{SVD} \\
	\cmidrule(lr){4-5} \cmidrule(lr){6-7}
	$\sigma$ & $n$ & $\mathrm{cond}(A)$ & $e_c$ & $e_n$ & $e_c$ & $e_n$ \\
	\midrule
	\num{e-3} & 16 & \num{5.39e5} & \num{2.31e-12} & \num{5.02e-16} & \num{8.97e-16} & \num{4.48e-16} \\
	& 32 & \num{2.11e6} & \num{1.04e-10} & \num{3.49e-16} & \num{2.47e-15} & \num{8.58e-16} \\
	& 64 & \num{8.32e6} & \num{1.83e-10} & \num{3.14e-16} & \num{2.52e-15} & \num{6.82e-16} \\
	& 128 & \num{3.30e7} & \num{4.72e-10} & \num{2.02e-16} & \num{4.93e-15} & \num{1.16e-15} \\
	\midrule
	\num{e-9} & 16 & \num{5.39e11} & \num{1.40e-6} & \num{5.56e-16} & \num{1.54e-15} & \num{3.56e-16} \\
	& 32 & \num{2.11e12} & \num{7.68e-5} & \num{2.93e-16} & \num{1.92e-15} & \num{9.45e-16} \\
	& 64 & \num{8.32e12} & \num{7.56e-5} & \num{2.72e-16} & \num{2.73e-15} & \num{1.02e-15} \\
	& 128 & \num{3.29e13} & \num{2.51e-3} & \num{3.74e-16} & \num{4.16e-15} & \num{1.19e-15} \\
	\midrule
	\num{e-12} & 16 & \num{5.42e14} & \num{6.52e-3} & \num{4.42e-16} & \num{1.93e-15} & \num{3.48e-16} \\
	& 32 & \num{2.17e15} & \num{2.98e-2} & \num{3.52e-16} & \num{2.29e-15} & \num{6.22e-16} \\
	& 64 & \num{5.22e15} & \num{2.53} & \num{2.04e-16} & \num{3.69e-15} & \num{9.03e-16} \\
	& 128 & \num{3.27e16} & \num{13.21} & \num{1.95e-16} & \num{4.16e-15} & \num{1.04e-15} \\
	\bottomrule
	\end{tabular}
    \label{tab_svd2}
\end{table}
\end{example}
\begin{example}\label{ex5}
 This example illustrates the accuracy of the \lstinline{sqrtm} function implemented in the \mmatrix-toolbox for computing the matrix square root of a triplet M-matrix. We consider the  matrix $A$ with its triplets from Example~\ref{ex1} where $n = 32$.

Figure~\ref{fig5} reports the normwise and componentwise relative errors obtained by \lstinline{sqrtm} from the \mmatrix-toolbox with $\texttt{maxit}=40$. The results show that the algorithm achieves componentwise accuracy for all values of $\sigma$, with both $E_n$ and $E_c$ remaining close to machine precision. For comparison, we also report  $\text{cond}(\sqrt{\cdot};A)\epsilon$, where $\text{cond}(\sqrt{\cdot};A)$ denotes the condition number of the matrix square root function at $A$ in the Frobenius norm. As one can see the errors obtained by \matlab's \lstinline{sqrtm} function (based on the Schur form of $A$) follow the conditioning of the matrix function while \mmatrix-toolbox's \lstinline{sqrtm} get componentwise accuracy also for matrices with ill-conditioned square root.
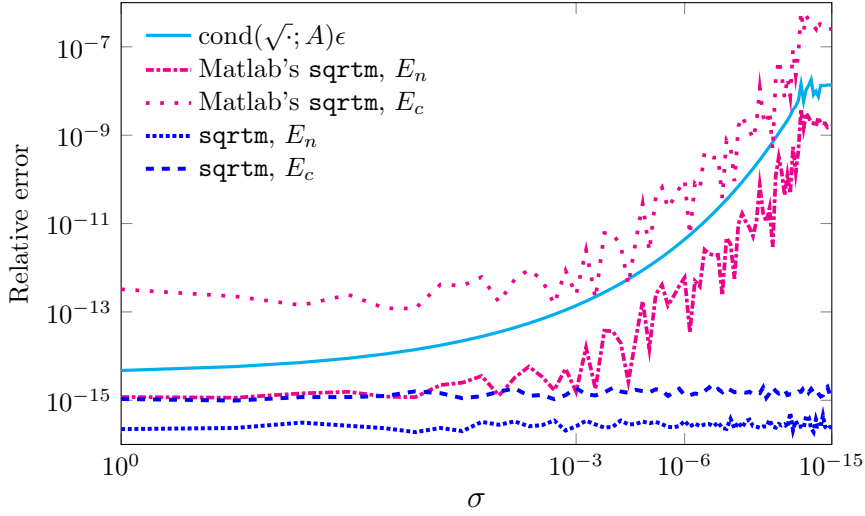
\begin{figure}[htbp]
\centering
\input{figsqrtm}
\caption{Example~\ref{ex5}: Accuracy comparison for \texttt{sqrtm}.}
\label{fig5}
\end{figure}
\end{example}
	
\section{Conclusions}\label{sec6}

We have presented the \mmatrix-toolbox, a software package for accurate computation in problems involving M-matrices with triplet representation.

The algorithms based on this technology provide striking numerical results, getting full componentwise accuracy beyond the theory of conditioning and timing comparable to \lapack{} implementations. The theoretical fundations are perturbation results, which we collect in a systematic way.

The price to pay is that one needs the triplet representation, but the theory guarantees that it exists under mild assumptions, while in practice, in most problems (Markov chains, complex networks, \ldots), it is known or easily obtained.

The package is easy-to-use due to the \matlab{} interface, efficient due to the Fortran implementation of the core routines, and flexible enough to be extended with future algorithms.

\appendix

\section{On Theorem~\ref{thm:componentwise perturbation bounds}}
A result that can be used to prove Theorem~\ref{thm:componentwise perturbation bounds} is the following.
\begin{lemma} \label{lem:matrixtree}
    Consider an M-matrix $B\in\mathbb{R}^{n\times n}$ with right triplet $(B, \mathbf{1}, v)$. Define the nonnegative quantities $w_{i,j} = -B_{i,j}$ for each $i,j=1,\dots,n$ with $i\neq j$, and $w_{i,0} = v_i$ for each $i=1,\dots,n$. Then, 
    \begin{enumerate}
        \item $\det(B)$ is the sum of nonnegative terms of the form
        \begin{equation} \label{detBterm}
        w_{1,\tau_1} w_{2,\tau_2} \dotsm w_{n,\tau_n},    
        \end{equation}
        where $\tau_i \in \{0,1,\dots,n\} \setminus \{i\}$ for each $i$.
        \item For each $i,j$, $\adj(B)_{ij}$ is the sum of nonnegative terms of the form
        \begin{equation} \label{eq:AdjBijterm}
        w_{1,\tau_1} w_{2,\tau_2} \dotsm w_{j-1,\tau_{j-1}}w_{j+1,\tau_{j+1}} \dotsm    w_{n,\tau_n}    
        \end{equation}
        (without the factor $w_{j,\tau_j}$).
    \end{enumerate}
\end{lemma}
Here, $\adj(B)$ is the adjugate matrix of $B$, i.e., the unique matrix such that $B \adj(B) = \adj(B) B = \det(B) I$.

Lemma~\ref{lem:matrixtree} appears in this form in~\cite[Theorem~3.1]{NKP}, but equivalent results using recursive expressions for cofactors are in~\cite[Section~2]{alfa}. It is also noted in~\cite{NKP} that the lemma is a special case of a more general combinatorial result (all-minors matrix tree theorem, \cite{chaiken}), and that therefore the terms have a combinatorial interpretation in terms of spanning trees in the complete directed graph on vertices $\{0,1,\dots,n\}$ with edge weights $w_{i,j}$. Here, we only report an illustrative example from~\cite{NKP}.

\begin{example}
    Let $n=3$, so that
    \[
    B = \begin{bmatrix}
            w_{1,0}+w_{1,2}+w_{1,3} & -w_{1,2} & -w_{1,3}\\ -w_{2,1} & w_{2,0}+w_{2,1}+w_{2,3} & -w_{2,3}\\ -w_{3,1} & -w_{3,2} & w_{3,0}+w_{3,1}+w_{3,2} 
    \end{bmatrix}.
    \]
    Then, 
    \begin{align*}
    \det B &=     w_{1,0}\,w_{2,0}\,w_{3,0}+w_{1,0}\,w_{2,0}\,w_{3,1}+w_{1,0}\,w_{2,1}\,w_{3,0}+w_{1,0}\,w_{2,0}\,w_{3,2} \\&+ w_{1,0}\,w_{2,1}\,w_{3,1} +w_{1,2}\,w_{2,0}\,w_{3,0}+w_{1,0}\,w_{2,1}\,w_{3,2}+w_{1,0}\,w_{2,3}\,w_{3,0}\\&+w_{1,2}\,w_{2,0}\,w_{3,1}+w_{1,3}\,w_{2,0}\,w_{3,0}+w_{1,0}\,w_{2,3}\,w_{3,1}+w_{1,2}\,w_{2,0}\,w_{3,2}\\&+w_{1,3}\,w_{2,1}\,w_{3,0}+w_{1,2}\,w_{2,3}\,w_{3,0}+w_{1,3}\,w_{2,0}\,w_{3,2}+w_{1,3}\,w_{2,3}\,w_{3,0},
    \end{align*}
    and
    \begin{align*}
    \adj(B)_{21} &= - \det \begin{bmatrix}
        -w_{2,1} & -w_{2,3}\\ -w_{3,1} &  w_{3,0}+w_{3,1}+w_{3,2}
    \end{bmatrix}\\ &= w_{2,1}\,w_{3,0}+w_{2,1}\,w_{3,1}+w_{2,1}\,w_{3,2}+w_{2,3}\,w_{3,1}.    
    \end{align*}
\end{example}
With this lemma, we prove Theorem~\ref{thm:componentwise perturbation bounds}, using mostly ideas from~\cite{alfa}.
\begin{proof}[Proof of Theorem~\ref{thm:componentwise perturbation bounds}]
It is enough to prove the result for right triplets; the arguments for left triplets are analogous, or follow immediately by transposing.
    \begin{enumerate}        
        \item The assumptions imply for all $i$ and $j\neq i$
        \[
        (1-\varepsilon)v_i \leq \tilde{v}_i \leq (1+\varepsilon)v_i, \quad (1-\varepsilon)(-A_{ij}) \leq (-\tilde{A}_{ij}) \leq (1+\varepsilon)(-A_{ij}).
        \]
        Using formula~\eqref{eq:1}, we get
        \begin{align*}
        \tilde{A}_{ii} &= \frac{\tilde{v}_i - \sum_{j\neq i} \tilde{A}_{ij}u_j}{u_{i}} \leq
        \frac{(1+\varepsilon)v_i + \sum_{j\neq i} (1+\varepsilon)(-A_{ij})u_j}{u_{i}} = (1+\varepsilon)A_{ii},    \\
        \tilde{A}_{ii} &= \frac{\tilde{v}_i - \sum_{j\neq i} \tilde{A}_{ij}u_j}{u_{i}} \geq
        \frac{(1-\varepsilon)v_i + \sum_{j\neq i} (1-\varepsilon)(-A_{ij})u_j}{u_{i}} = (1-\varepsilon)A_{ii}.
        \end{align*}
        \item We define $B = A \operatorname{diag}(u)$, which satisfies $B\boldsymbol{1}=v$, and $\tilde{B}$ analogously. The assumptions imply
        \[
        (1-\varepsilon)w_{i,j} \leq \tilde{w}_{i,j} \leq (1+\varepsilon) w_{i,j},
        \]
        where the nonnegative weights $\tilde{w}_{i,j}$ (equal to either $-\tilde{B}_{ij} = -\tilde{A}_{ij}u_j$ or $\tilde{v}_i$ when $j=0$) are defined as in Lemma~\ref{lem:matrixtree} starting from $\tilde{B}$,         
        so each summand~\eqref{detBterm} of $\det(B)$ satisfies
        \[
        (1-\varepsilon)^n w_{1,\tau_1} w_{2,\tau_2} \dotsm w_{n,\tau_n} \leq
        \tilde{w}_{1,\tau_1} \tilde{w}_{2,\tau_2} \dotsm \tilde{w}_{n,\tau_n}
        \leq
        (1+\varepsilon)^n
        w_{1,\tau_1} w_{2,\tau_2} \dotsm w_{n,\tau_n},
        \]
        and hence 
        \[
        (1-\varepsilon)^n \det(B) \leq \det(\tilde{B}) \leq (1+\varepsilon)^n \det(B).
        \]
        Since $\det(A) = \frac{\det(B)}{\det(\operatorname{diag}(u))}$, the same inequalities apply to $\det(A)$.
        \item The bounds in the previous item ensure that $\det(\tilde{A})\neq 0$ whenever $\det(A)\neq 0$. Arguing as in the previous item with the second part of Lemma~\ref{lem:matrixtree}, we get the inequality
        $(1-\varepsilon)^{n-1}\adj(A)_{ij} \leq \adj(\tilde{A})_{ij} \leq (1+\varepsilon)^{n-1} \adj(A)_{ij}$, since this time each product~\eqref{eq:AdjBijterm} has $n-1$ factors rather than $n$. Now
        \[
        (\tilde{A}^{-1})_{ij} = \frac{\adj(\tilde{A})_{ij}}{\det(\tilde{A})} \leq \frac{(1+\varepsilon)^{n-1}\adj(A)_{ij}}{(1-\varepsilon)^n \det(A)} = \frac{(1+\varepsilon)^{n-1}}{(1-\varepsilon)^n}(A^{-1})_{ij}
        \]
        and analogously for the lower bound.
        \item Note that $A$ and $\tilde{A}$ have the same zero pattern, hence they have the same Frobenius normal form. We reduce without loss of generality to the case in which they are both irreducible. The singular case is easy: if $\lambda=0$, then $v=0$, and we must have $\tilde{v}=0$ as well for $e(\tilde{v},v) \leq \varepsilon$ to hold.
        
        We now deal with the non-singular case. Note that $\tilde{\lambda}^{-1}$ and $\lambda^{-1}$ are the Perron eigenvalues of the non-negative matrices $\tilde{A}^{-1}$ and $A^{-1}$, respectively. The monotonicity of the Perron vector then gives immediately
        \[
        \frac{(1-\varepsilon)^{n-1}}{(1+\varepsilon)^{n}}\lambda^{-1} \leq  \tilde{\lambda}^{-1} \leq \frac{(1+\varepsilon)^{n-1}}{(1-\varepsilon)^{n}}\lambda^{-1},
        \]
        from which the result follows.
        \item If $A$ is singular, $A \adj(A) = 0$, so each column of $\adj(A)$ is a multiple of the Perron eigenvector. Let us take a column, for instance $j=1$, and normalize it to obtain an expression for the entries of $z$:
        \[
        z_{i} = \frac{\adj(A)_{i1}}{\adj(A)_{11}+\dots+\adj(A)_{n1}}.
        \]
        We now argue as in the previous items to get the bound
        \[
        \frac{(1-\varepsilon)^{n-1}}{(1+\varepsilon)^{n-1}} z_i \leq \tilde{z}_i \leq \frac{(1+\varepsilon)^{n-1}}{(1-\varepsilon)^{n-1}}z_i.
        \]
    \end{enumerate}
\end{proof}
This proof strategy extends to a finer bound that allows perturbations of $u$ as well.
\begin{theorem} \label{thm:generalized componentwise perturbation bounds}
    Let $\tilde{A}, A\in\mathbb{R}^{n\times n}$ be M-matrices with (left or right) triplets $(A,u,v)$ and $(\tilde{A}, \tilde{u}, \tilde{v})$. Suppose that $e(-\offdiag(\tilde{A}),-\offdiag(A)) \leq \varepsilon_o$, $e(\tilde{u},u) \leq \varepsilon_u$ and $e(\tilde{v},v) \leq \varepsilon_v$, for some positive constants $\varepsilon_o, \varepsilon_u, \varepsilon_v<1$. Let $\varepsilon_{\max},\varepsilon_{\min}$ be defined so that
    \[
    1+\varepsilon_{\max} = \max(1+\varepsilon_v, (1+\varepsilon_o)(1+\varepsilon_u)), \quad 
    1-\varepsilon_{\min} = \min(1-\varepsilon_v, (1-\varepsilon_o)(1-\varepsilon_u)).
    \]
    Then,
    \begin{enumerate}
        \item \leavevmode \vspace{-1ex}
        \[
        \frac{1-\varepsilon_{\min}}{1+\varepsilon_u}\operatorname{diag}(A) \leq  \operatorname{diag}(\tilde{A}) \leq \frac{1+\varepsilon_{\max}}{1-\varepsilon_u}\operatorname{diag}(A);
        \]
        \item \leavevmode \vspace{-1ex}
        \[\frac{(1-\varepsilon_{\min})^n}{(1+\varepsilon_u)^n}\operatorname{det}(A) \leq  \operatorname{det}(\tilde{A}) \leq \frac{(1+\varepsilon_{\max})^n}{(1-\varepsilon_u)^n} \operatorname{det}(A);
        \]
        \item If $A$ is invertible, then so is $\tilde{A}$, and
        \[
        \frac{(1-\varepsilon_{\min})^{n-1}}{(1+\varepsilon_u)^{n-1}}\frac{(1-\varepsilon_u)^n}{(1+\varepsilon_{\max})^n}A^{-1} \leq  \tilde{A}^{-1} \leq \frac{(1+\varepsilon_{\max})^{n-1}}{(1-\varepsilon_u)^{n-1}} \frac{(1+\varepsilon_u)^{n}}{(1-\varepsilon_{\min})^n}A^{-1};
        \]
        \item If $\tilde{\lambda}$ and $\lambda$ are the Perron eigenvalues of $\tilde{A}$ and $A$, respectively, then
        \[
        \frac{(1-\varepsilon_u)^{n-1}}{(1+\varepsilon_{\max})^{n-1}} \frac{(1-\varepsilon_{\min})^n}{(1+\varepsilon_u)^{n}}
        \lambda \leq  \tilde{\lambda} \leq
        \frac{(1+\varepsilon_u)^{n-1}}{(1-\varepsilon_{\min})^{n-1}}\frac{(1+\varepsilon_{\max})^n}{(1-\varepsilon_u)^n}
        \lambda;
        \]
        \item If $\tilde{A}$ and $A$ are singular irreducible, and $\tilde{z},z\in\mathbb{R}^{n}$  are the positive Perron eigenvectors, normalized such that $\norm{\tilde{z}}_1=\norm{z}_1=1$, then
        \[
        \frac{(1-\varepsilon_{\min})^{n-1}}{(1+\varepsilon_u)^{n-1}}
        \frac{(1-\varepsilon_{u})^{n-1}}{(1+\varepsilon_{\max})^{n-1}}
        z \leq  \tilde{z} \leq 
        \frac{(1+\varepsilon_{\max})^{n-1}}{(1-\varepsilon_u)^{n-1}}
        \frac{(1+\varepsilon_u)^{n-1}}{(1-\varepsilon_{\min})^{n-1}}
        z.
        \]
    \end{enumerate}
\end{theorem}
The bounds in Items 2--4 can be marginally improved by replacing the $n$th powers $(1+\varepsilon_{\max})^n$ and $(1-\varepsilon_{\min})^n$ with $(1+\varepsilon_{\max})^{n-1}(1+\varepsilon_v)$ and $(1-\varepsilon_{\min})^{n-1}(1-\varepsilon_v)$: this follows from noting that each product~\eqref{detBterm} contains at least one $\tau_k$ equal to zero, because $\det(B)$ must vanish when $v=0$.

We can convert also the results of Theorem~\ref{thm:generalized componentwise perturbation bounds} into simpler-looking first-order asymptotic bounds for $\varepsilon \rightarrow 0$: for instance,
\[
e(\tilde{A}^{-1}, A^{-1}) \leq (2n-1)\max(\varepsilon_o+\varepsilon_u,\varepsilon_v) + (2n-1)\varepsilon_u + \mathcal{O}(\max(\varepsilon_o,\varepsilon_u,\varepsilon_v)^2).
\]

\bibliographystyle{siamplain}
\bibliography{references}

\end{document}

%% file: ex_shared.tex
\usepackage{lipsum}
\usepackage{amsfonts}
\usepackage{graphicx}
\usepackage{epstopdf}
\usepackage{algorithmic}
\ifpdf
  \DeclareGraphicsExtensions{.eps,.pdf,.png,.jpg}
\else
  \DeclareGraphicsExtensions{.eps}
\fi

\usepackage{amsthm}
\newtheorem{theorem}{Theorem}
\newtheorem{lemma}[theorem]{Lemma}
\usepackage[ruled]{algorithm2e}
\usepackage{hyperref}

\title{The \texttt{mmatrix} toolbox:  componentwise accurate algorithms for M-matrices with triplet representation\thanks{
This research has been supported by ICSC--Centro Nazionale di Ricerca in High Performance Computing, Big Data, and Quantum Computing funded by European Union--NextGenerationEU; by PRIN 2022ZK5ME7 CUP B53C24006410006 funded by Italian MUR; and by the project “AIDMIX” funded by the University of Perugia; BI and FP are members of INDAM (Istituto Nazionale di Alta Matematica).  
}}

\def\email{\texttt}

\author{Bruno Iannazzo\thanks{Department of Mathematics and Computer Science, University of Perugia, Italy
  (\email{bruno.iannazzo@unipg.it}).}
\and
Mehdi Najafi Kalyani\thanks{Department of Computer Science, University of Pisa, Italy
(\email{mehdi.najafi@di.unipi.it}, \email{federico.poloni@unipi.it}).}
\and Federico Poloni\footnotemark[3]}

\usepackage{amsopn}
\DeclareMathOperator{\diag}{diag}


%% file: fig1.tex
\definecolor{mycolor1}{rgb}{0.00000,1.00000,1.00000}%
\definecolor{mycolor2}{rgb}{0.00000,0.50000,0.80000}%
\definecolor{mycolor3}{rgb}{0.00000,0.70000,0.60000}%
\definecolor{mycolor4}{rgb}{0.00000,0.90000,0.40000}%
\definecolor{mycolor5}{rgb}{0.60000,0.00000,0.80000}%
\definecolor{mycolor6}{rgb}{0.80000,0.00000,0.80000}%
\definecolor{mycolor7}{rgb}{1.00000,0.00000,0.80000}%
\definecolor{mycolor8}{rgb}{1.00000,0.20000,1.00000}%
\definecolor{marker4}{RGB}{231,41,138}
\begin{tikzpicture}[scale=.95]

\begin{axis}[%
width=4.55in,
height=3in,
scale only axis,
xmin=1,
xmax=23,
   xmode=log,
    xtick={1,3,7,23},
xticklabels={
	$10^{-1}$, $10^{-2}$, $10^{-4}$, $10^{-12}$
},
xminorticks=true,
xlabel={$\sigma$},
ymode=log,
ymin=9.9e-17,
ymax=101,
yminorticks=true,
ylabel={Relative error},
legend style={
	at={(0.02,0.999)},
	anchor=north west,
	legend cell align=left,
	align=left,
	draw=none,
	font=\small,
},
grid style={dotted,gray!50},
legend columns=3
]
\addplot [color=mycolor1, line width=2.5pt]
  table[row sep=crcr]{%
1	1.18576860155564e-09\\
2	3.70784035173105e-09\\
3	1.16834550096887e-08\\
4	3.69046025823762e-08\\
5	1.16660888735313e-07\\
6	3.68872404148935e-07\\
7	1.16643504440362e-06\\
8	3.68854769042246e-06\\
9	1.16641550229821e-05\\
10	3.68854057167281e-05\\
11	0.000116642773804194\\
12	0.00036887553471791\\
13	0.00116668912648949\\
14	0.00368827776787956\\
15	0.0116689757177773\\
16	0.0370262114294834\\
17	0.114474716191214\\
18	0.345116628844039\\
19	1.29139790920928\\
20	3.32501404827248\\
21	8.37718225019507\\
22	4.6903260735271\\
23	145.569273471098\\
};
\addlegendentry{$\text{cond}(A)\epsilon$}

\addplot [color=marker4, line width=1pt, dashed]
  table[row sep=crcr]{%
1	8.07760695837856e-13\\
2	4.68204862317032e-11\\
3	7.52770084362455e-11\\
4	2.79238916733177e-10\\
5	1.15898902298907e-09\\
6	1.91879287438443e-08\\
7	3.25180910097216e-08\\
8	1.01810565946214e-07\\
9	8.42239054453502e-07\\
10	1.58052166668595e-06\\
11	5.38824450504801e-07\\
12	2.71109647880666e-07\\
13	4.32305492244054e-07\\
14	1.20999509365708e-06\\
15	5.30799124966939e-06\\
16	1.68150367141609e-05\\
17	6.07840580318416e-06\\
18	5.28401147574938e-05\\
19	0.000119972481731028\\
20	0.00138608469352035\\
21	0.00239242704500598\\
22	0.0085190412173722\\
23	0.024561484217402\\
};
\addlegendentry{LU, $e_{n}$}

\addplot [color=marker4, dashed, only marks, line width=1.2pt, mark=x,mark size=3pt, mark options={solid, marker4}]
  table[row sep=crcr]{%
1	9.31092868303147e-13\\
2	4.6902965399428e-11\\
3	7.55494532300913e-11\\
4	2.79352337043426e-10\\
5	1.15905136897112e-09\\
6	1.91895972666023e-08\\
7	3.2518280942736e-08\\
8	1.0181067214456e-07\\
9	8.42239361071688e-07\\
10	1.58052185121731e-06\\
11	5.38824478397054e-07\\
12	2.71109658284972e-07\\
13	4.3230550677953e-07\\
14	1.20999510975496e-06\\
15	5.3079912718743e-06\\
16	1.68150367400245e-05\\
17	6.07840581270423e-06\\
18	5.28401147628029e-05\\
19	0.000119972481742667\\
20	0.00138608469353203\\
21	0.00239242704501203\\
22	0.00851904121738795\\
23	0.0245614842174145\\
};
\addlegendentry{LU, $e_{c}$}

\addplot [color=red!20!orange, line width=1.0pt, mark=square, mark options={solid, red!20!orange}]
  table[row sep=crcr]{%
1	2.25327600896003e-15\\
2	2.37690465569224e-15\\
3	2.37833947312234e-15\\
4	2.79531711035773e-15\\
5	5.90020870484298e-16\\
6	1.11022502060284e-15\\
7	5.96575112401328e-16\\
8	2.00009466561666e-15\\
9	4.12025895282362e-16\\
10	2.00922120609418e-15\\
11	1.24354785999384e-15\\
12	7.48301389386806e-15\\
13	4.33694100620048e-14\\
14	8.28539864304952e-15\\
15	2.39810043506745e-14\\
16	2.7648813796723e-14\\
17	9.93050990660139e-15\\
18	3.86759294542254e-15\\
19	2.47758506214232e-14\\
20	2.7814526079675e-14\\
21	1.01061239253884e-14\\
22	6.75839645834853e-15\\
23	1.49290132458912e-14\\
};
\addlegendentry{GTH, $e_{n}$}

\addplot [color=red!20!orange, dashed, line width=1.2pt, mark=square, mark options={solid, red!20!orange}]
  table[row sep=crcr]{%
1	4.30572357733263e-15\\
2	3.87112812457005e-15\\
3	4.12127491074813e-15\\
4	4.79024471632182e-15\\
5	1.60408069036979e-15\\
6	2.36714451737397e-15\\
7	1.283279164278e-15\\
8	3.06611010569835e-15\\
9	9.12558785502503e-16\\
10	3.60721570319836e-15\\
11	2.92019573700331e-15\\
12	1.65066188530112e-14\\
13	6.87706312991746e-14\\
14	1.58832929671701e-14\\
15	3.45751284323663e-14\\
16	4.71327484328118e-14\\
17	1.73810105092436e-14\\
18	5.67365999942866e-15\\
19	3.34911514328651e-14\\
20	3.06377639969525e-14\\
21	1.10042354707962e-14\\
22	1.42219743985854e-14\\
23	2.21998489497849e-14\\
};
\addlegendentry{GTH, $e_{c}$}

\addplot [color=blue!70!mycolor1, line width=1.0pt, mark=o, mark options={solid, blue!70!mycolor1}]
  table[row sep=crcr]{%
1	2.25327600896003e-15\\
2	2.37690465569224e-15\\
3	2.18354614227193e-15\\
4	2.79531711035773e-15\\
5	5.90020870484298e-16\\
6	1.11022502060284e-15\\
7	6.8349351650584e-16\\
8	2.00009466561666e-15\\
9	4.12025895282362e-16\\
10	2.00922120609418e-15\\
11	1.24354785999384e-15\\
12	7.48301389386806e-15\\
13	4.33694100620048e-14\\
14	8.28539864304952e-15\\
15	2.38687812652229e-14\\
16	2.7648813796723e-14\\
17	9.93050990660139e-15\\
18	3.86759294542254e-15\\
19	2.46294391628629e-14\\
20	2.76257019573057e-14\\
21	9.98671581819194e-15\\
22	6.75839645834853e-15\\
23	1.49290132458912e-14\\
};
\addlegendentry{\texttt{dmmtf2}, $e_{n}$}

\addplot [color=blue!70!mycolor1, dashed, line width=1.2pt, mark=o, mark options={solid, blue!70!mycolor1}]
  table[row sep=crcr]{%
1	4.30572357733263e-15\\
2	3.87112812457005e-15\\
3	3.89850329395093e-15\\
4	4.79024471632182e-15\\
5	1.60408069036979e-15\\
6	2.36714451737397e-15\\
7	1.42586573808667e-15\\
8	3.06611010569835e-15\\
9	9.12558785502503e-16\\
10	3.60721570319836e-15\\
11	2.92019573700331e-15\\
12	1.65066188530112e-14\\
13	6.87706312991746e-14\\
14	1.58832929671701e-14\\
15	3.44583205660408e-14\\
16	4.71327484328118e-14\\
17	1.73810105092436e-14\\
18	5.67365999942866e-15\\
19	3.33416373639684e-14\\
20	3.04486419969714e-14\\
21	1.08846242156788e-14\\
22	1.42219743985854e-14\\
23	2.21998489497849e-14\\
};
\addlegendentry{\texttt{dmmtf2}, $e_{c}$}

\addplot [color=mycolor2, line width=1.0pt, mark=square, mark options={solid, mycolor2}]
  table[row sep=crcr]{%
1	2.25327600896003e-15\\
2	2.37690465569224e-15\\
3	2.18354614227193e-15\\
4	2.79531711035773e-15\\
5	5.90020870484298e-16\\
6	1.11022502060284e-15\\
7	6.8349351650584e-16\\
8	2.00009466561666e-15\\
9	4.12025895282362e-16\\
10	2.00922120609418e-15\\
11	1.24354785999384e-15\\
12	7.48301389386806e-15\\
13	4.33694100620048e-14\\
14	8.28539864304952e-15\\
15	2.38687812652229e-14\\
16	2.7648813796723e-14\\
17	9.93050990660139e-15\\
18	3.86759294542254e-15\\
19	2.46294391628629e-14\\
20	2.76257019573057e-14\\
21	9.98671581819194e-15\\
22	6.75839645834853e-15\\
23	1.49290132458912e-14\\
};
\addlegendentry{\texttt{dmmtrf2}, $e_{n}$}

\addplot [color=mycolor2, dashed, line width=1.2pt, mark=square, mark options={solid, mycolor2}]
  table[row sep=crcr]{%
1	4.30572357733263e-15\\
2	3.87112812457005e-15\\
3	3.89850329395093e-15\\
4	4.79024471632182e-15\\
5	1.60408069036979e-15\\
6	2.36714451737397e-15\\
7	1.42586573808667e-15\\
8	3.06611010569835e-15\\
9	9.12558785502503e-16\\
10	3.60721570319836e-15\\
11	2.92019573700331e-15\\
12	1.65066188530112e-14\\
13	6.87706312991746e-14\\
14	1.58832929671701e-14\\
15	3.44583205660408e-14\\
16	4.71327484328118e-14\\
17	1.73810105092436e-14\\
18	5.67365999942866e-15\\
19	3.33416373639684e-14\\
20	3.04486419969714e-14\\
21	1.08846242156788e-14\\
22	1.42219743985854e-14\\
23	2.21998489497849e-14\\
};
\addlegendentry{\texttt{dmmtrf2}, $e_{c}$}

\addplot [color=mycolor3, line width=1.0pt, mark=triangle, mark options={solid, mycolor3}]
  table[row sep=crcr]{%
1	2.25327600896003e-15\\
2	2.37690465569224e-15\\
3	2.18354614227193e-15\\
4	2.79531711035773e-15\\
5	5.90020870484298e-16\\
6	1.11022502060284e-15\\
7	6.8349351650584e-16\\
8	2.00009466561666e-15\\
9	4.12025895282362e-16\\
10	2.00922120609418e-15\\
11	1.24354785999384e-15\\
12	7.48301389386806e-15\\
13	4.33694100620048e-14\\
14	8.28539864304952e-15\\
15	2.38687812652229e-14\\
16	2.7648813796723e-14\\
17	9.93050990660139e-15\\
18	3.86759294542254e-15\\
19	2.46294391628629e-14\\
20	2.76257019573057e-14\\
21	9.98671581819194e-15\\
22	6.75839645834853e-15\\
23	1.49290132458912e-14\\
};
\addlegendentry{\texttt{dmmtrf}, $e_{n}$}

\addplot [color=mycolor3, dashed, line width=1.2pt, mark=triangle, mark options={solid, mycolor3}]
  table[row sep=crcr]{%
1	4.30572357733263e-15\\
2	3.87112812457005e-15\\
3	3.89850329395093e-15\\
4	4.79024471632182e-15\\
5	1.60408069036979e-15\\
6	2.36714451737397e-15\\
7	1.42586573808667e-15\\
8	3.06611010569835e-15\\
9	9.12558785502503e-16\\
10	3.60721570319836e-15\\
11	2.92019573700331e-15\\
12	1.65066188530112e-14\\
13	6.87706312991746e-14\\
14	1.58832929671701e-14\\
15	3.44583205660408e-14\\
16	4.71327484328118e-14\\
17	1.73810105092436e-14\\
18	5.67365999942866e-15\\
19	3.33416373639684e-14\\
20	3.04486419969714e-14\\
21	1.08846242156788e-14\\
22	1.42219743985854e-14\\
23	2.21998489497849e-14\\
};
\addlegendentry{\texttt{dmmtrf}, $e_{c}$}

\addplot [color=mycolor4, line width=1.0pt, mark=diamond, mark options={solid, mycolor4}]
  table[row sep=crcr]{%
1	2.25327600896003e-15\\
2	2.37690465569224e-15\\
3	2.18354614227193e-15\\
4	2.79531711035773e-15\\
5	5.90020870484298e-16\\
6	1.11022502060284e-15\\
7	6.8349351650584e-16\\
8	2.00009466561666e-15\\
9	4.12025895282362e-16\\
10	2.00922120609418e-15\\
11	1.24354785999384e-15\\
12	7.48301389386806e-15\\
13	4.33694100620048e-14\\
14	8.28539864304952e-15\\
15	2.38687812652229e-14\\
16	2.7648813796723e-14\\
17	9.93050990660139e-15\\
18	3.86759294542254e-15\\
19	2.46294391628629e-14\\
20	2.76257019573057e-14\\
21	9.98671581819194e-15\\
22	6.75839645834853e-15\\
23	1.49290132458912e-14\\
};
\addlegendentry{\texttt{dmmtrff}, $e_{n}$}

\addplot [color=mycolor4, dashed, line width=1.2pt, mark=diamond, mark options={solid, mycolor4}]
  table[row sep=crcr]{%
1	4.30572357733263e-15\\
2	3.87112812457005e-15\\
3	3.89850329395093e-15\\
4	4.79024471632182e-15\\
5	1.60408069036979e-15\\
6	2.36714451737397e-15\\
7	1.42586573808667e-15\\
8	3.06611010569835e-15\\
9	9.12558785502503e-16\\
10	3.60721570319836e-15\\
11	2.92019573700331e-15\\
12	1.65066188530112e-14\\
13	6.87706312991746e-14\\
14	1.58832929671701e-14\\
15	3.44583205660408e-14\\
16	4.71327484328118e-14\\
17	1.73810105092436e-14\\
18	5.67365999942866e-15\\
19	3.33416373639684e-14\\
20	3.04486419969714e-14\\
21	1.08846242156788e-14\\
22	1.42219743985854e-14\\
23	2.21998489497849e-14\\
};
\addlegendentry{\texttt{dmmtrff}, $e_{c}$}

\addplot [color=mycolor5, line width=1.0pt, mark=o, mark options={solid, mycolor5}]
  table[row sep=crcr]{%
1	2.52911764549905e-15\\
2	2.21049105997947e-15\\
3	2.98590468114248e-15\\
4	3.65262409412519e-15\\
5	1.34903155924852e-15\\
6	8.7306978335423e-16\\
7	7.18066381823486e-15\\
8	4.63212440575463e-15\\
9	1.18619639446013e-14\\
10	8.18860326994274e-15\\
11	1.05725870120107e-15\\
12	7.40340919027028e-15\\
13	4.33694100620048e-14\\
14	8.28539864304952e-15\\
15	2.39810043506745e-14\\
16	2.7648813796723e-14\\
17	9.93050990660139e-15\\
18	3.86759294542254e-15\\
19	2.47758506214232e-14\\
20	2.7814526079675e-14\\
21	1.01061239253884e-14\\
22	6.75839645834853e-15\\
23	1.49290132458912e-14\\
};
\addlegendentry{\texttt{dmltf2}, $e_{n}$}

\addplot [color=mycolor5, dashed, line width=1.2pt, mark=o, mark options={solid, mycolor5}]
  table[row sep=crcr]{%
1	4.58351219522506e-15\\
2	3.69516775527141e-15\\
3	4.78958976113972e-15\\
4	5.7764715696822e-15\\
5	2.49523662946411e-15\\
6	1.69088244421353e-15\\
7	8.12743470709402e-15\\
8	6.3126277917305e-15\\
9	1.2775822997035e-14\\
10	9.37876088605213e-15\\
11	2.73768355621605e-15\\
12	1.63911879519412e-14\\
13	6.87706312991746e-14\\
14	1.58832929671701e-14\\
15	3.45751284323663e-14\\
16	4.71327484328118e-14\\
17	1.73810105092436e-14\\
18	5.67365999942866e-15\\
19	3.34911514328651e-14\\
20	3.06377639969525e-14\\
21	1.10042354707962e-14\\
22	1.42219743985854e-14\\
23	2.21998489497849e-14\\
};
\addlegendentry{\texttt{dmltf2}, $e_{c}$}

\addplot [color=mycolor6, line width=1.0pt, mark=square, mark options={solid, mycolor6}]
  table[row sep=crcr]{%
1	2.52911764549905e-15\\
2	2.21049105997947e-15\\
3	2.98590468114248e-15\\
4	3.65262409412519e-15\\
5	1.34903155924852e-15\\
6	8.7306978335423e-16\\
7	7.18066381823486e-15\\
8	4.63212440575463e-15\\
9	1.18619639446013e-14\\
10	8.18860326994274e-15\\
11	1.05725870120107e-15\\
12	7.40340919027028e-15\\
13	4.33694100620048e-14\\
14	8.28539864304952e-15\\
15	2.39810043506745e-14\\
16	2.7648813796723e-14\\
17	9.93050990660139e-15\\
18	3.86759294542254e-15\\
19	2.47758506214232e-14\\
20	2.7814526079675e-14\\
21	1.01061239253884e-14\\
22	6.75839645834853e-15\\
23	1.49290132458912e-14\\
};
\addlegendentry{\texttt{dmltrf2}, $e_{n}$}

\addplot [color=mycolor6, dashed, line width=1.2pt, mark=square, mark options={solid, mycolor6}]
  table[row sep=crcr]{%
1	4.58351219522506e-15\\
2	3.69516775527141e-15\\
3	4.78958976113972e-15\\
4	5.7764715696822e-15\\
5	2.49523662946411e-15\\
6	1.69088244421353e-15\\
7	8.12743470709402e-15\\
8	6.3126277917305e-15\\
9	1.2775822997035e-14\\
10	9.37876088605213e-15\\
11	2.73768355621605e-15\\
12	1.63911879519412e-14\\
13	6.87706312991746e-14\\
14	1.58832929671701e-14\\
15	3.45751284323663e-14\\
16	4.71327484328118e-14\\
17	1.73810105092436e-14\\
18	5.67365999942866e-15\\
19	3.34911514328651e-14\\
20	3.06377639969525e-14\\
21	1.10042354707962e-14\\
22	1.42219743985854e-14\\
23	2.21998489497849e-14\\
};
\addlegendentry{\texttt{dmltrf2}, $e_{c}$}

\addplot [color=mycolor7, line width=1.0pt, mark=triangle, mark options={solid, mycolor7}]
  table[row sep=crcr]{%
1	2.52911764549905e-15\\
2	2.21049105997947e-15\\
3	2.98590468114248e-15\\
4	3.65262409412519e-15\\
5	1.34903155924852e-15\\
6	8.7306978335423e-16\\
7	7.18066381823486e-15\\
8	4.63212440575463e-15\\
9	1.18619639446013e-14\\
10	8.18860326994274e-15\\
11	1.05725870120107e-15\\
12	7.40340919027028e-15\\
13	4.33694100620048e-14\\
14	8.28539864304952e-15\\
15	2.39810043506745e-14\\
16	2.7648813796723e-14\\
17	9.93050990660139e-15\\
18	3.86759294542254e-15\\
19	2.47758506214232e-14\\
20	2.7814526079675e-14\\
21	1.01061239253884e-14\\
22	6.75839645834853e-15\\
23	1.49290132458912e-14\\
};
\addlegendentry{\texttt{dmltrf}, $e_{n}$}

\addplot [color=mycolor7, dashed, line width=1.2pt, mark=triangle, mark options={solid, mycolor7}]
  table[row sep=crcr]{%
1	4.58351219522506e-15\\
2	3.69516775527141e-15\\
3	4.78958976113972e-15\\
4	5.7764715696822e-15\\
5	2.49523662946411e-15\\
6	1.69088244421353e-15\\
7	8.12743470709402e-15\\
8	6.3126277917305e-15\\
9	1.2775822997035e-14\\
10	9.37876088605213e-15\\
11	2.73768355621605e-15\\
12	1.63911879519412e-14\\
13	6.87706312991746e-14\\
14	1.58832929671701e-14\\
15	3.45751284323663e-14\\
16	4.71327484328118e-14\\
17	1.73810105092436e-14\\
18	5.67365999942866e-15\\
19	3.34911514328651e-14\\
20	3.06377639969525e-14\\
21	1.10042354707962e-14\\
22	1.42219743985854e-14\\
23	2.21998489497849e-14\\
};
\addlegendentry{\texttt{dmltrf}, $e_{c}$}

\addplot [color=mycolor8, line width=1.0pt, mark=diamond, mark options={solid, mycolor8}]
  table[row sep=crcr]{%
1	2.52911764549905e-15\\
2	2.21049105997947e-15\\
3	2.98590468114248e-15\\
4	3.65262409412519e-15\\
5	1.34903155924852e-15\\
6	8.7306978335423e-16\\
7	7.18066381823486e-15\\
8	4.63212440575463e-15\\
9	1.18619639446013e-14\\
10	8.18860326994274e-15\\
11	1.05725870120107e-15\\
12	7.40340919027028e-15\\
13	4.33694100620048e-14\\
14	8.28539864304952e-15\\
15	2.39810043506745e-14\\
16	2.7648813796723e-14\\
17	9.93050990660139e-15\\
18	3.86759294542254e-15\\
19	2.47758506214232e-14\\
20	2.7814526079675e-14\\
21	1.01061239253884e-14\\
22	6.75839645834853e-15\\
23	1.49290132458912e-14\\
};
\addlegendentry{\texttt{dmltrff}, $e_{n}$}

\addplot [color=mycolor8, dashed, line width=1.2pt, mark=diamond, mark options={solid, mycolor8}]
  table[row sep=crcr]{%
1	4.58351219522506e-15\\
2	3.69516775527141e-15\\
3	4.78958976113972e-15\\
4	5.7764715696822e-15\\
5	2.49523662946411e-15\\
6	1.69088244421353e-15\\
7	8.12743470709402e-15\\
8	6.3126277917305e-15\\
9	1.2775822997035e-14\\
10	9.37876088605213e-15\\
11	2.73768355621605e-15\\
12	1.63911879519412e-14\\
13	6.87706312991746e-14\\
14	1.58832929671701e-14\\
15	3.45751284323663e-14\\
16	4.71327484328118e-14\\
17	1.73810105092436e-14\\
18	5.67365999942866e-15\\
19	3.34911514328651e-14\\
20	3.06377639969525e-14\\
21	1.10042354707962e-14\\
22	1.42219743985854e-14\\
23	2.21998489497849e-14\\
};
\addlegendentry{\texttt{dmltrff}, $e_{c}$}

\end{axis}

\end{tikzpicture}%

%% file: fig3.tex
\definecolor{mycolor6}{RGB}{255,140,0}       
 \definecolor{marker1}{RGB}{27,158,119}  
 \definecolor{marker4}{RGB}{231,41,138} 
 \definecolor{mycolor6}{RGB}{124,104,238}    
\begin{tikzpicture}[scale=.9]

\begin{axis}[%
width=4.5in,
height=2.3in,
xtick={1,7,13,25},
xticklabels={$10^{0}$, $10^{-3}$, $10^{-6}$, $10^{-12}$},
xminorticks=true,
ymode=log,
ylabel={Relative error},
scale only axis,
xmode=log,
xmin=1,
xmax=26,
xminorticks=true,
xlabel={$\sigma$},
xlabel style={font=\large},
ymode=log,
ymin=1e-16,
ymax=1000,
yminorticks=true,
legend style={
	at={(0.02,0.98)},
	anchor=north west,
	legend cell align=left,
	align=left,
	draw=none,
	font=\normalsize,
},
]
\addplot [color=cyan, line width=2.0pt]
  table[row sep=crcr]{%
1	8.55172995878347e-12\\
2	2.43812980864819e-11\\
3	7.45051389902892e-11\\
4	2.33034080264857e-10\\
5	7.34354286021136e-10\\
6	2.31967043103604e-09\\
7	7.33288102601495e-09\\
8	2.31860450071982e-08\\
9	7.33181504667446e-08\\
10	2.3184976806846e-07\\
11	7.33170525821069e-07\\
12	2.3184865986896e-06\\
13	7.33167398079921e-06\\
14	2.31850991604963e-05\\
15	7.33183411572722e-05\\
16	0.00023182540695531\\
17	0.000732855715206304\\
18	0.00231962056303154\\
19	0.00731335356045103\\
20	0.0231415341140232\\
21	0.0714841060143342\\
22	0.213380610183238\\
23	0.462320596539639\\
24	1.17511392664299\\
25	7.2706850476608\\
26	1519.99677134352\\
};
\addlegendentry{$\text{cond}(A)\epsilon$}

\addplot [color=marker1, loosely dashed, only marks,line width=.8pt, mark=o, mark options={solid, fill=marker1, marker1,rotate=90},mark size=4pt]  table[row sep=crcr]{%
1	1.44010361994727e-13\\
2	2.39835856252652e-13\\
3	5.23415980306573e-13\\
4	8.14062519584898e-13\\
5	6.07849507177777e-12\\
6	1.48288073801358e-11\\
7	1.21024649756665e-10\\
8	2.68688118788437e-11\\
9	1.15832639405343e-09\\
10	3.54512934144603e-08\\
11	1.09343007479789e-08\\
12	4.05122463044487e-07\\
13	4.70238716787712e-07\\
14	1.38671715923164e-07\\
15	1.9191719722082e-07\\
16	4.51472801716691e-07\\
17	7.11698238494841e-07\\
18	2.27992694012957e-06\\
19	1.07412582882824e-05\\
20	2.28132790375976e-06\\
21	4.65487441270037e-05\\
22	2.28146799996949e-06\\
23	0.00033309507748205\\
24	0.00363865341467522\\
25	0.00252499947546507\\
26	0.00363865341607794\\
};
\addlegendentry{\matlab's \texttt{inv}, $E_n$}

\addplot [color=marker1, loosely dashed,only marks, line width=.5pt, mark=*, mark options={solid, fill=marker1, marker1,rotate=90},mark size=2.3pt]
  table[row sep=crcr]{%
1	1.84126218545526e-13\\
2	2.71022670000077e-13\\
3	5.41500308935775e-13\\
4	8.23972260243425e-13\\
5	6.12933040468846e-12\\
6	1.49049454518081e-11\\
7	1.21071038223959e-10\\
8	2.69186972826536e-11\\
9	1.15843558894668e-09\\
10	3.54516561464678e-08\\
11	1.09343365978153e-08\\
12	4.05122879457048e-07\\
13	4.7023886304667e-07\\
14	1.386717293847e-07\\
15	1.91917212270933e-07\\
16	4.51472811309972e-07\\
17	7.11698243147689e-07\\
18	2.27992694484308e-06\\
19	1.07412582953047e-05\\
20	2.2813279042314e-06\\
21	4.65487441300469e-05\\
22	2.28146800001666e-06\\
23	0.000333095077484228\\
24	0.00363865341468275\\
25	0.00252499947546672\\
26	0.00363865341607869\\
};
\addlegendentry{\matlab's \texttt{inv}, $E_c$}

\addplot [
		color=marker4,only marks,
		mark=diamond*,
		mark size=4.5pt,
		mark options={rotate=45}
		]
  table[row sep=crcr]{%
1	5.2475599721952e-16\\
2	4.86668927415618e-16\\
3	1.14109771185148e-15\\
4	9.60143180487891e-16\\
5	8.59222086998119e-16\\
6	8.34676176364213e-16\\
7	5.64068777357531e-16\\
8	1.00792618030048e-15\\
9	1.70942017536848e-15\\
10	2.15889137618764e-15\\
11	1.27429672565654e-15\\
12	3.8618399148368e-15\\
13	4.26212333623382e-15\\
14	7.48376558170013e-16\\
15	3.67288487358962e-15\\
16	3.8150127808032e-15\\
17	2.83541551926945e-15\\
18	2.07473784982869e-15\\
19	7.9886845525943e-15\\
20	1.60531563543829e-15\\
21	1.33279641518061e-15\\
22	3.33307423140805e-15\\
23	1.11516774638388e-15\\
24	3.85001460219981e-15\\
25	2.44569522627628e-15\\
26	1.76681311889563e-15\\
};
\addlegendentry{\texttt{inv}, $E_n$}

\addplot [
		color=marker4,only marks,
		mark=diamond,line width=.8pt,
		mark size=6pt,
		mark options={rotate=45}
		]
  table[row sep=crcr]{%
1	2.2148131370511e-15\\
2	1.61925930545715e-15\\
3	3.42276371898567e-15\\
4	3.08271106584657e-15\\
5	2.80984296923077e-15\\
6	2.49769906642114e-15\\
7	2.14258316006285e-15\\
8	2.28309924286487e-15\\
9	4.5120101664634e-15\\
10	5.82162409504001e-15\\
11	3.89854072647506e-15\\
12	1.04105733947574e-14\\
13	1.03961551444121e-14\\
14	2.92226990889042e-15\\
15	1.07195945432515e-14\\
16	9.58504655000306e-15\\
17	5.47068990398847e-15\\
18	3.74050603478112e-15\\
19	1.86890594628444e-14\\
20	1.94506313536122e-15\\
21	2.27107558150714e-15\\
22	3.35149340199904e-15\\
23	2.27107558140653e-15\\
24	7.66055634731905e-15\\
25	3.63372093023435e-15\\
26	2.29816690419249e-15\\
};
\addlegendentry{\texttt{inv}, $E_c$}

\addplot [color=mycolor6, dotted,only marks, line width=.8pt, mark=square, mark options={solid, fill=mycolor6, mycolor6,rotate=45},mark size=4pt]
  table[row sep=crcr]{%
1	5.40582238550155e-16\\
2	5.43179769881993e-16\\
3	1.17253409086333e-15\\
4	1.06157020454577e-15\\
5	8.81179664777312e-16\\
6	9.02900770442059e-16\\
7	6.24425879441699e-16\\
8	1.06608522986589e-15\\
9	1.77279582475055e-15\\
10	2.15591624944061e-15\\
11	1.21703603166975e-15\\
12	3.8166815846336e-15\\
13	4.35990523529417e-15\\
14	7.38596060803386e-16\\
15	6.14322631793129e-15\\
16	3.8150127808032e-15\\
17	2.83541551926945e-15\\
18	2.07473784982869e-15\\
19	7.9886845525943e-15\\
20	1.60531563543829e-15\\
21	1.33279641518061e-15\\
22	3.33307423140805e-15\\
23	1.11516774638388e-15\\
24	3.85001460219981e-15\\
25	2.44569522627628e-15\\
26	1.76681311889563e-15\\
};
\addlegendentry{\texttt{inv}*, $E_n$}

\addplot [color=mycolor6, dotted,only marks, line width=.8pt, mark=square*, mark options={solid, mycolor6,rotate=45},mark size=2.5pt]
  table[row sep=crcr]{%
1	1.94757788190385e-15\\
2	2.0498043600362e-15\\
3	4.08251609740211e-15\\
4	3.7578901877502e-15\\
5	3.08820561460559e-15\\
6	2.4959857073575e-15\\
7	2.2551411303428e-15\\
8	3.13894601235474e-15\\
9	4.69249784797404e-15\\
10	5.36502218121559e-15\\
11	3.03219881292441e-15\\
12	8.9494396285694e-15\\
13	1.15512834042465e-14\\
14	2.92226993038118e-15\\
15	1.66338536015971e-14\\
16	9.58504655000306e-15\\
17	5.47068990398847e-15\\
18	3.74050603478112e-15\\
19	1.86890594628444e-14\\
20	1.94506313536122e-15\\
21	2.27107558150714e-15\\
22	3.35149340199904e-15\\
23	2.27107558140653e-15\\
24	7.66055634731905e-15\\
25	3.63372093023435e-15\\
26	2.29816690419249e-15\\
};
\addlegendentry{\texttt{inv}*, $E_c$}

\end{axis}
 
\end{tikzpicture}%

%% file: figsvd.tex
\definecolor{mycolor1}{rgb}{1.00000,0.00000,1.00000}%
\definecolor{mycolor2}{rgb}{0.52157,0.08627,0.81961}%
\definecolor{mycolor3}{rgb}{0.12941,0.12941,0.12941}%

\begin{tikzpicture}[scale=.9]
	\begin{groupplot}[
		group style={
			group size=2 by 2,
			horizontal sep=2cm,
			vertical sep=1.65cm,
		},
		width=0.55\textwidth,
		height=0.45\textwidth,
		ymode=log,
		yminorticks=true,
		xlabel={$i$},
		ylabel={$\sigma_i$},
		enlargelimits=false,
		]
		
		\nextgroupplot[title={$n = 20,\,\,\kappa(\diag(u)) = \num{2.29e+02}$}]
		\addplot [color=red, line width=1.0pt, only marks, mark size=5.0pt, mark=x, mark options={solid, red}, forget plot]
		table[row sep=crcr]{%
1	1.19389642254509e+93\\
2	8.66601997239223e+76\\
3	4.14912665932095e+43\\
4	1.02631313334027e+43\\
5	4.14049840692451e+31\\
6	1.12721221499694e+28\\
7	1.44021645143139e+25\\
8	3.89492812335651e+21\\
9	2.88913491674012e+18\\
10	93422689559.7459\\
11	9.30542789625786\\
12	7.67565140661716e-24\\
13	4.13313533315534e-26\\
14	4.29978157889732e-50\\
15	3.0777384695605e-51\\
16	4.63217963562902e-57\\
17	3.96455465003753e-70\\
18	2.0314422990531e-80\\
19	9.45905063217673e-81\\
20	1.52900848890499e-111\\
		};
		\addplot [color=cyan, line width=1.0pt, only marks, mark size=4.2pt, mark=square, mark options={solid, cyan}, forget plot]
		table[row sep=crcr]{%
1	1.19389642254509e+93\\
2	8.66601997239222e+76\\
3	4.14912665932095e+43\\
4	1.02631313334027e+43\\
5	4.1404984069245e+31\\
6	1.12721221499694e+28\\
7	1.44021645143139e+25\\
8	3.89492812335651e+21\\
9	2.88913491674012e+18\\
10	93422689559.7459\\
11	9.30542789625785\\
12	7.67565140661716e-24\\
13	4.13313533315534e-26\\
14	4.29978157889732e-50\\
15	3.0777384695605e-51\\
16	4.63217963562902e-57\\
17	3.96455465003754e-70\\
18	2.0314422990531e-80\\
19	9.45905063217673e-81\\
20	1.52900848890499e-111\\
		};
		\addplot [color=mycolor1, only marks, mark size=4.0pt, mark=o, mark options={solid, mycolor1}, forget plot]
		table[row sep=crcr]{%
1	1.19389642254509e+93\\
2	8.66601997239222e+76\\
3	4.14912665932095e+43\\
4	1.02631313334027e+43\\
5	4.1404984069248e+31\\
6	1.12721593091332e+28\\
7	2.30259357191137e+26\\
8	1.9652721687727e+25\\
9	3.88122569617548e+21\\
10	170201564399.574\\
11	13522619784.431\\
12	1327816083.49952\\
13	28385.5968583211\\
14	1.03814365760732e-06\\
15	4.53090832540794e-24\\
16	3.02435215099323e-26\\
17	4.6920544593352e-40\\
18	1.79426091758488e-50\\
19	1.98054064214276e-58\\
20	0\\
		};
		\addplot [color=mycolor2, line width=1.0pt, forget plot]
		table[row sep=crcr]{%
0	2.65098259465433e+77\\
20	2.65098259465433e+77\\
		};
		\node[right, align=left, inner sep=0, font=\bfseries\color{mycolor3}]
		at (axis cs:10.56,00000000900000999999990006912290463574917186943833581808988099134357504) {$\sigma_{1} \epsilon=\num{2.7e+77}$};
		
		\nextgroupplot[title={$n = 30,\,\,\kappa(\diag(u)) = \num{9.16e+02}$}]
		\addplot [color=red, line width=1.0pt, only marks, mark size=5.0pt, mark=x, mark options={solid, red}, forget plot]
		table[row sep=crcr]{%
1	1.27291727265108e+94\\
2	1.53918214737242e+83\\
3	1.80487688666141e+81\\
4	7.93474809573993e+70\\
5	1.9306861311089e+58\\
6	5.1284863300409e+55\\
7	3.30099874220313e+41\\
8	3.66729097869853e+39\\
9	1.99146647770948e+28\\
10	3.74371915439574e+25\\
11	7.84132602845096e+23\\
12	6.5621493103641e+22\\
13	4.15004343209391e+20\\
14	213372963463882\\
15	4391268195.07101\\
16	0.0640942096709736\\
17	2.56302880107142e-08\\
18	1.37577556220638e-10\\
19	5.0127402705442e-20\\
20	6.32449461712872e-21\\
21	5.67502726011892e-29\\
22	6.62469327132979e-30\\
23	4.09220307350486e-54\\
24	2.10676758905553e-58\\
25	5.95566317999363e-63\\
26	1.46428212682702e-74\\
27	3.32480031135954e-78\\
28	1.53337918291096e-82\\
29	8.15986647925028e-93\\
30	8.06495195846472e-120\\
		};
		\addplot [color=cyan, line width=1.0pt, only marks, mark size=4.2pt, mark=square, mark options={solid, cyan}, forget plot]
		table[row sep=crcr]{%
1	1.27291727265108e+94\\
2	1.53918214737242e+83\\
3	1.80487688666141e+81\\
4	7.93474809573994e+70\\
5	1.9306861311089e+58\\
6	5.1284863300409e+55\\
7	3.30099874220313e+41\\
8	3.66729097869854e+39\\
9	1.99146647770948e+28\\
10	3.74371915439574e+25\\
11	7.84132602845096e+23\\
12	6.5621493103641e+22\\
13	4.15004343209391e+20\\
14	213372963463882\\
15	4391268195.07102\\
16	0.0640942096709736\\
17	2.56302880107142e-08\\
18	1.37577556220638e-10\\
19	5.0127402705442e-20\\
20	6.32449461712872e-21\\
21	5.67502726011892e-29\\
22	6.62469327132979e-30\\
23	4.09220307350486e-54\\
24	2.10676758905553e-58\\
25	5.95566317999364e-63\\
26	1.46428212682702e-74\\
27	3.32480031135954e-78\\
28	1.53337918291096e-82\\
29	8.15986647925027e-93\\
30	8.06495195846489e-120\\
		};
		\addplot [color=mycolor1, only marks, mark size=4.0pt, mark=o, mark options={solid, mycolor1}, forget plot]
		table[row sep=crcr]{%
1	1.27291727265108e+94\\
2	1.53918214737242e+83\\
3	1.80487732453705e+81\\
4	1.25907697571738e+78\\
5	1.25907697571738e+78\\
6	1.25907697571738e+78\\
7	1.25907697571738e+78\\
8	1.25907697571738e+78\\
9	1.25907697571738e+78\\
10	1.25907697571738e+78\\
11	1.25907697571738e+78\\
12	1.25907697571738e+78\\
13	1.25907697571738e+78\\
14	1.25907697571738e+78\\
15	1.25907697571738e+78\\
16	1.25907697571738e+78\\
17	1.25907697571738e+78\\
18	1.25907697571738e+78\\
19	1.25907697571738e+78\\
20	1.25907697571738e+78\\
21	1.25907697571738e+78\\
22	1.25907697571738e+78\\
23	1.25907697571738e+78\\
24	1.25907697571738e+78\\
25	1.25907697571738e+78\\
26	1.25907697571738e+78\\
27	1.25907697571738e+78\\
28	1.25907697571738e+78\\
29	1.25907697571738e+78\\
30	6.82007473128856e+76\\
		};
		\addplot [color=mycolor2, line width=1.0pt, forget plot]
		table[row sep=crcr]{%
0	2.82644412908057e+78\\
30	2.82644412908057e+78\\
		};
		\node[right, align=left, inner sep=0, font=\bfseries\color{mycolor3}]
		at (axis cs:15.12,0000000000100993986441619330569931716029349942514092427408236871680) {$\sigma_{1} \epsilon=\num{2.8e+78}$};
		
		\nextgroupplot[title={$n = 20,\,\,\kappa(\diag(u)) = \num{1.35e+24}$}]
		\addplot [color=red, line width=1.0pt, only marks, mark size=5.0pt, mark=x, mark options={solid, red}, forget plot]
		table[row sep=crcr]{%
1	2.9479584683823e+105\\
2	4.88794020723556e+94\\
3	2.43209016583583e+92\\
4	4.59485406329169e+87\\
5	1.86307357017366e+85\\
6	5.17962823153198e+82\\
7	8.75925335154088e+61\\
8	3.20450655350087e+58\\
9	2.77672296979128e+52\\
10	1.35636101525612e+48\\
11	8.78103430763313e+44\\
12	1.51938922665277e+38\\
13	9.07695389364965e-10\\
14	1.24335866471219e-17\\
15	3.52038050200187e-23\\
16	9.55610490442483e-36\\
17	3.18870369998566e-41\\
18	1.30583420507848e-74\\
19	1.18607617102416e-83\\
20	8.79907869524052e-132\\
		};
		\addplot [color=cyan, line width=1.0pt, only marks, mark size=4.2pt, mark=square, mark options={solid, cyan}, forget plot]
		table[row sep=crcr]{%
1	2.9479584683823e+105\\
2	4.88794020723556e+94\\
3	2.43209016583583e+92\\
4	4.59485406329169e+87\\
5	1.86307357017366e+85\\
6	5.17962823153198e+82\\
7	8.75925335154089e+61\\
8	3.20450655350087e+58\\
9	2.77672296979128e+52\\
10	1.35636101525612e+48\\
11	8.78103430763314e+44\\
12	1.51938922665277e+38\\
13	9.07695389364965e-10\\
14	1.24335866471219e-17\\
15	3.52038050200187e-23\\
16	9.55610490442483e-36\\
17	3.18870369998566e-41\\
18	1.30583420507848e-74\\
19	1.18607617102416e-83\\
20	8.79907869524046e-132\\
		};
		\addplot [color=mycolor1, only marks, mark size=4.0pt, mark=o, mark options={solid, mycolor1}, forget plot]
		table[row sep=crcr]{%
1	2.9479584683823e+105\\
2	4.88794020723556e+94\\
3	2.43209016583583e+92\\
4	4.59485406329169e+87\\
5	1.8630721132126e+85\\
6	5.17962823153198e+82\\
7	4.08258930241221e+77\\
8	1.32206041555187e+70\\
9	1.44357718542007e+69\\
10	2.77281919302791e+63\\
11	2.33044037688965e+61\\
12	1.09498825606726e+59\\
13	4.46794758470126e+56\\
14	2.7054176941294e+52\\
15	1.5864674393565e+45\\
16	6.32508906715533e+42\\
17	2.8140225444638e+40\\
18	1.54215089449369e+37\\
19	3.35797392894033e+23\\
20	7.79288092597098e-75\\
		};
		\addplot [color=mycolor2, line width=1.0pt, forget plot]
		table[row sep=crcr]{%
0	6.54578273447347e+89\\
20	6.54578273447347e+89\\
		};
		\node[right, align=left, inner sep=0, font=\bfseries\color{mycolor3}]
		at (axis cs:10.9,99999999999999999999999951781162173498456687995536454521311154393764421697536) {$\sigma_{1} \epsilon=\num{6.5e+89}$};
		
		\nextgroupplot[title={$n = 30,\,\,\kappa(\diag(u)) = \num{3.37e+28}$},
		legend style={at={(-1.4,-.3)}, anchor=south west, legend cell align=left, align=left, 
			legend columns=4, draw=none, /tikz/every even column/.append style={column sep=0.5cm}}]
		\addplot [color=red, line width=1.0pt, only marks, mark size=5.0pt, mark=x, mark options={solid, red}]
		table[row sep=crcr]{%
1	1.30128260296463e+118\\
2	9.92502423895837e+103\\
3	1.12727399408468e+90\\
4	2.49531046362798e+86\\
5	4.71238692190421e+79\\
6	2.80502904662052e+78\\
7	1.89383365581482e+75\\
8	7.51761391138994e+73\\
9	3.3583022211817e+71\\
10	8.40936138108835e+66\\
11	3.72830965176947e+53\\
12	1.40911954697495e+43\\
13	2.0608658339722e+27\\
14	1.779597786682e+24\\
15	1.24964133513787e+24\\
16	2.52718971122839e+21\\
17	5813746.49497506\\
18	824710.349289117\\
19	818.495829870991\\
20	0.000884711691827161\\
21	1.34147911329802e-13\\
22	8.32158030644285e-35\\
23	3.1234694845199e-39\\
24	2.84633314157372e-42\\
25	1.04168695765599e-42\\
26	5.6654772645036e-45\\
27	4.35657051513062e-48\\
28	6.60742972385186e-61\\
29	1.09442469142573e-66\\
30	3.26839255140699e-125\\
		};
		\addlegendentry{SVD (from \mmatrix-toolbox)}
		
		\addplot [color=cyan, line width=1.0pt, only marks, mark size=4.2pt, mark=square, mark options={solid, cyan}]
		table[row sep=crcr]{%
1	1.30128260296463e+118\\
2	9.92502423895837e+103\\
3	1.12727399408468e+90\\
4	2.49531046362798e+86\\
5	4.71238692190422e+79\\
6	2.80502904662052e+78\\
7	1.89383365581482e+75\\
8	7.51761391138993e+73\\
9	3.3583022211817e+71\\
10	8.40936138108835e+66\\
11	3.72830965176947e+53\\
12	1.40911954697495e+43\\
13	2.0608658339722e+27\\
14	1.779597786682e+24\\
15	1.24964133513787e+24\\
16	2.52718971122839e+21\\
17	5813746.49497506\\
18	824710.349289117\\
19	818.495829870991\\
20	0.00088471169182716\\
21	1.34147911329802e-13\\
22	8.32158030644289e-35\\
23	3.1234694845199e-39\\
24	2.84633314157372e-42\\
25	1.04168695765599e-42\\
26	5.6654772645036e-45\\
27	4.35657051513062e-48\\
28	6.60742972385187e-61\\
29	1.09442469142573e-66\\
30	3.27966459303519e-125\\
		};
		\addlegendentry{SVD via \texttt{vpa}}
		
		\addplot [color=mycolor1, only marks, mark size=4.0pt, mark=o, mark options={solid, mycolor1}]
		table[row sep=crcr]{%
1	1.30128260296463e+118\\
2	9.92587477892244e+103\\
3	1.27645298359189e+102\\
4	1.27645298359189e+102\\
5	1.27645298359189e+102\\
6	1.27645298359189e+102\\
7	1.27645298359189e+102\\
8	1.27645298359189e+102\\
9	1.27645298359189e+102\\
10	1.27645298359189e+102\\
11	1.27645298359189e+102\\
12	1.27645298359189e+102\\
13	1.27645298359189e+102\\
14	1.27645298359189e+102\\
15	1.27645298359189e+102\\
16	1.27645298359189e+102\\
17	1.27645298359189e+102\\
18	1.27645298359189e+102\\
19	1.27645298359189e+102\\
20	1.27645298359189e+102\\
21	1.27645298359189e+102\\
22	1.27645298359189e+102\\
24	1.27645298359189e+102\\
25	1.27645298359189e+102\\
26	1.27645298359189e+102\\
27	1.27645298359189e+102\\
28	1.27645298359189e+102\\
29	1.27645298359189e+102\\
30	3.1256547014152e+99\\
		};
		\addlegendentry{\matlab{}'s SVD}
		
		\addplot [color=mycolor2, line width=1.5pt]
		table[row sep=crcr]{%
0	2.88942781471099e+102\\
30	2.88942781471099e+102\\
		};
		\addlegendentry{$\sigma{}_\text{1}\text{ }\epsilon$}
		
		\node[right, align=left, inner sep=0, font=\bfseries\color{mycolor3}]
		at (axis cs:15.120,999999999999999999999999999999999999999997292411887728373640844069515128963072) {$\sigma_{1}\epsilon=\num{2.9e+102}$};
		
	\end{groupplot}
\end{tikzpicture}

%% file: figsqrtm.tex
\begin{tikzpicture} 
	
\begin{axis}[%
width=3.7in,
height=2.3in,
scale only axis,
xmode=log,
xmin=1,
xmax=76,
xtick={1,16,31,76},
xticklabels={$10^{0}$, $10^{-3}$, $10^{-6}$, $10^{-15}$},
xlabel={$\sigma$},
xlabel style={font=\large},
yminorticks=false,
ymode=log,
ymin=1e-16,
ymax=1e-6,
ylabel={Relative error},
legend style={
	at={(0.02,0.98)},
	anchor=north west,
	legend cell align=left,
	align=left,
	draw=none,
	font=\normalsize,
},
grid style={dotted,gray!50},
legend columns=1
]
\addplot [color=cyan, line width=1.2pt]
  table[row sep=crcr]{%
1	4.71522449555562e-15\\
2	5.78783640119793e-15\\
3	7.16815452747027e-15\\
4	8.93001200332408e-15\\
5	1.11673537538338e-14\\
6	1.39993702782275e-14\\
7	1.75768932498253e-14\\
8	2.2090445221345e-14\\
9	2.77803919041109e-14\\
10	3.49497448215196e-14\\
11	4.39802986554619e-14\\
12	5.53529635566714e-14\\
13	6.9673375509298e-14\\
14	8.77041487219086e-14\\
15	1.10405487641372e-13\\
16	1.38986321365477e-13\\
17	1.74968683491653e-13\\
18	2.20268766078684e-13\\
19	2.77298963529606e-13\\
20	3.49096342392795e-13\\
21	4.39484373822145e-13\\
22	5.53276551220733e-13\\
23	6.96532722893029e-13\\
24	8.76881795432923e-13\\
25	1.1039280322344e-12\\
26	1.38976243826136e-12\\
27	1.74960675069923e-12\\
28	2.20262406032155e-12\\
29	2.77293921152078e-12\\
30	3.49092350597495e-12\\
31	4.39481235445462e-12\\
32	5.53274024739054e-12\\
33	6.9653061195514e-12\\
34	8.76880084244527e-12\\
35	1.10392804979554e-11\\
36	1.38976309219603e-11\\
37	1.74960843390721e-11\\
38	2.20261953834483e-11\\
39	2.77293788216199e-11\\
40	3.49089905536022e-11\\
41	4.39478445703218e-11\\
42	5.53274711794454e-11\\
43	6.9655023933478e-11\\
44	8.76827284377909e-11\\
45	1.10387188612447e-10\\
46	1.38969684580869e-10\\
47	1.74967113141551e-10\\
48	2.20239869351451e-10\\
49	2.77064040496038e-10\\
50	3.4917836923637e-10\\
51	4.39393718035087e-10\\
52	5.53557846509127e-10\\
53	6.96471372942728e-10\\
54	8.80740929970946e-10\\
55	1.12002454119849e-09\\
56	1.40142015102743e-09\\
57	1.74413940494788e-09\\
58	2.1775162230054e-09\\
59	2.78823073871943e-09\\
60	3.78774193076925e-09\\
61	4.51161352031675e-09\\
62	6.1623523856351e-09\\
63	1.40889993727577e-08\\
64	5.77986948615697e-09\\
65	1.07674378136555e-08\\
66	1.32971568967602e-08\\
67	1.81750942094562e-08\\
68	8.18904048061008e-09\\
69	9.88457396010311e-09\\
70	7.51915625056602e-09\\
71	1.33641285097883e-08\\
72	1.33641285097883e-08\\
73	1.33641285097883e-08\\
74	1.36656922631231e-08\\
75	1.36656922631231e-08\\
76	1.36656922631231e-08\\
};
\addlegendentry{$\text{cond}(\sqrt{\cdot};A)\epsilon$}

\addplot [color=magenta , line width=1.4pt,densely dashdotted]
  table[row sep=crcr]{%
1	1.18652555391595e-15\\
2	1.14208013823547e-15\\
3	1.43473995188573e-15\\
4	1.55491051605305e-15\\
5	1.21110950561546e-15\\
6	1.17257239968318e-15\\
7	2.21263188830384e-15\\
8	2.49959673883952e-15\\
9	3.54376364720069e-15\\
10	1.28346054724462e-15\\
11	3.22457399816205e-15\\
12	5.78439992235035e-15\\
13	3.40888113762823e-15\\
14	1.69477722867477e-15\\
15	5.20768702614091e-15\\
16	1.55243252407951e-15\\
17	1.56840619132932e-14\\
18	1.80124131164453e-15\\
19	3.7568726702201e-14\\
20	3.65766283155153e-14\\
21	1.97388826466277e-14\\
22	2.65188937663831e-15\\
23	2.42637724168803e-14\\
24	1.74449985567543e-13\\
25	1.46627422097448e-14\\
26	2.51132662843546e-13\\
27	4.0856498397779e-13\\
28	1.40265146054291e-13\\
29	2.28792360137172e-13\\
30	3.89442161792249e-13\\
31	5.99092854436582e-13\\
32	3.45837819795282e-14\\
33	5.0124833096615e-13\\
34	3.78584460178843e-13\\
35	2.5617668696177e-12\\
36	2.06439561054338e-12\\
37	1.92983184227944e-12\\
38	1.95973458286092e-12\\
39	2.375485907753e-13\\
40	4.75458514668068e-12\\
41	3.40845492850089e-12\\
42	5.68714665925685e-13\\
43	9.78640020802318e-12\\
44	1.67131966312952e-11\\
45	1.09362788172649e-11\\
46	8.12390320761463e-12\\
47	5.1684749004905e-12\\
48	1.12537231836713e-11\\
49	7.2660672498083e-11\\
50	1.71528202270273e-11\\
51	1.09984519925012e-11\\
52	2.25120805352318e-11\\
53	2.97013505782263e-12\\
54	1.21425467489903e-10\\
55	3.1624662375615e-10\\
56	1.45401658455892e-10\\
57	4.35133378997704e-11\\
58	1.2715237921543e-10\\
59	4.80425329072008e-11\\
60	5.45253604380014e-10\\
61	1.43088740497158e-10\\
62	4.48553205344449e-10\\
63	3.8901743124799e-09\\
64	1.43246906610765e-09\\
65	3.1512555163161e-09\\
66	2.5281605596911e-09\\
67	1.92704279924077e-09\\
68	1.86338332466897e-09\\
69	1.58139813965153e-09\\
70	2.53461074132134e-09\\
71	1.89972861793831e-09\\
72	1.86982241068192e-09\\
73	1.85070423226306e-09\\
74	1.50043880695629e-09\\
75	1.55741086056191e-09\\
76	1.60266536995578e-09\\
};
\addlegendentry{\matlab's \texttt{sqrtm}, $E_n$}

\addplot [color=magenta , dashed, line width=1.4pt,loosely dotted]
table[row sep=crcr]{%
1	3.2471767579043e-13\\
2	2.24864527993802e-13\\
3	1.41891310270872e-13\\
4	2.44342684161656e-13\\
5	1.20945248177494e-13\\
6	1.19260281106834e-13\\
7	4.17621251449039e-13\\
8	3.91987670120323e-13\\
9	6.16723700298419e-13\\
10	1.66643597836367e-13\\
11	5.40547777937919e-13\\
12	9.10662156951341e-13\\
13	5.34728904894892e-13\\
14	1.49125335473014e-13\\
15	8.62778291773486e-13\\
16	2.4540445009414e-13\\
17	2.5260098857769e-12\\
18	2.00981315650334e-13\\
19	6.11737286194034e-12\\
20	5.87970643876796e-12\\
21	3.1682658311427e-12\\
22	4.1761278882e-13\\
23	3.90197145841618e-12\\
24	2.79933823339084e-11\\
25	2.37438184357292e-12\\
26	4.03237597448391e-11\\
27	6.55348827094329e-11\\
28	2.24466883616433e-11\\
29	3.664975686139e-11\\
30	6.23706714395504e-11\\
31	9.60283517144491e-11\\
32	5.57363635790296e-12\\
33	8.03209499630756e-11\\
34	6.06311868902517e-11\\
35	4.10436311858065e-10\\
36	3.30749029872224e-10\\
37	3.09109471484637e-10\\
38	3.13910273600973e-10\\
39	3.8012577957167e-11\\
40	7.61628236697502e-10\\
41	5.45961049201028e-10\\
42	9.1129585199555e-11\\
43	1.56765803560556e-09\\
44	2.67719416351448e-09\\
45	1.75191266966771e-09\\
46	1.30133047528977e-09\\
47	8.27890754334718e-10\\
48	1.80266900786628e-09\\
49	1.16390651610238e-08\\
50	2.74766793127255e-09\\
51	1.76178648025249e-09\\
52	3.60605672489479e-09\\
53	4.7570591209047e-10\\
54	1.9450284916418e-08\\
55	5.06571971556807e-08\\
56	2.32908532867013e-08\\
57	6.97011266201224e-09\\
58	2.03676011348623e-08\\
59	7.69560024833188e-09\\
60	8.73399931638095e-08\\
61	2.29202280361655e-08\\
62	7.18503718904785e-08\\
63	6.23137522379769e-07\\
64	2.29456260798889e-07\\
65	5.04775672591294e-07\\
66	4.04966779553714e-07\\
67	3.08678288979494e-07\\
68	2.98481090533497e-07\\
69	2.53312052174851e-07\\
70	4.05999929850153e-07\\
71	3.04303013277516e-07\\
72	2.99512561740649e-07\\
73	2.96450162943549e-07\\
74	2.40343833050491e-07\\
75	2.49469748398773e-07\\
76	2.56718720385957e-07\\
};
\addlegendentry{\matlab's \texttt{sqrtm}, $E_c$}

\addplot [color=blue, line width=1.5pt, densely dotted]
  table[row sep=crcr]{%
1	2.22597753666128e-16\\
2	2.35806037981839e-16\\
3	3.11973345952711e-16\\
4	2.64931064238083e-16\\
5	2.3271636704463e-16\\
6	1.89899975208338e-16\\
7	2.36099212657482e-16\\
8	2.00753813804309e-16\\
9	3.24112540777736e-16\\
10	2.78821713317832e-16\\
11	3.32408991633717e-16\\
12	2.57972807456085e-16\\
13	2.733653331934e-16\\
14	3.48154558189479e-16\\
15	2.06513675337702e-16\\
16	2.58247228581378e-16\\
17	3.44051383051933e-16\\
18	2.8334125462653e-16\\
19	2.40680430257476e-16\\
20	2.50757441411768e-16\\
21	3.47149847635487e-16\\
22	2.84168712783988e-16\\
23	2.89195525783498e-16\\
24	3.30730597108232e-16\\
25	3.02967851852162e-16\\
26	3.23728938757048e-16\\
27	2.87425468038123e-16\\
28	2.49364460026488e-16\\
29	2.44763562241762e-16\\
30	2.61796663518987e-16\\
31	3.39541690613718e-16\\
32	2.5414891802063e-16\\
33	2.99503193161309e-16\\
34	2.583032567625e-16\\
35	2.82953867977681e-16\\
36	2.89532009871615e-16\\
37	2.58489172047743e-16\\
38	3.15169933530482e-16\\
39	2.34444535787942e-16\\
40	2.25237614347275e-16\\
41	3.80021188830036e-16\\
42	2.30235404268285e-16\\
43	3.18383748928601e-16\\
44	2.41729077128876e-16\\
45	3.34320891328662e-16\\
46	3.64391444165264e-16\\
47	2.61436117886046e-16\\
48	2.50608850650473e-16\\
49	2.39864128274291e-16\\
50	2.55983297771428e-16\\
51	2.53941466969837e-16\\
52	3.60252258795908e-16\\
53	2.50987197165621e-16\\
54	2.40012981756762e-16\\
55	2.90460807064568e-16\\
56	2.68294545388163e-16\\
57	2.8625319611234e-16\\
58	2.44583419983963e-16\\
59	5.15272948494572e-16\\
60	2.10009698097043e-16\\
61	3.43301198351648e-16\\
62	2.52138325955792e-16\\
63	2.56673031350214e-16\\
64	2.35071494778729e-16\\
65	4.3656908092622e-16\\
66	4.16209529053514e-16\\
67	2.33248475706033e-16\\
68	1.91921021918257e-16\\
69	3.71198301680799e-16\\
70	2.51655219873627e-16\\
71	2.54695681334833e-16\\
72	4.48847902620711e-16\\
73	2.4857217488693e-16\\
74	2.24017876881231e-16\\
75	2.46612409382523e-16\\
76	2.45765426052437e-16\\
};
\addlegendentry{\texttt{sqrtm}, $E_n$}

\addplot [color=blue, dashed, line width=1.5pt]
  table[row sep=crcr]{%
1	1.06930496710709e-15\\
2	9.86505497736009e-16\\
3	1.18048967148499e-15\\
4	1.18274630657171e-15\\
5	1.25188745943835e-15\\
6	1.59230736056182e-15\\
7	1.44771163743751e-15\\
8	1.1148747892075e-15\\
9	1.23781072312937e-15\\
10	1.5640538595052e-15\\
11	1.85721806689168e-15\\
12	1.26280505361487e-15\\
13	1.31456587689932e-15\\
14	1.06059553740948e-15\\
15	1.30639317806182e-15\\
16	1.8878482767078e-15\\
17	1.3501712195292e-15\\
18	1.2866031459319e-15\\
19	1.5054170592609e-15\\
20	1.9056077109787e-15\\
21	1.5776277730024e-15\\
22	1.475386879412e-15\\
23	1.7215722033554e-15\\
24	1.61119111471078e-15\\
25	1.51334635434539e-15\\
26	1.77543755316648e-15\\
27	1.68014455709781e-15\\
28	1.25428905663343e-15\\
29	1.50228500000128e-15\\
30	1.81928698356605e-15\\
31	1.29694493042752e-15\\
32	1.30064758939309e-15\\
33	1.42260713105853e-15\\
34	1.35248486054728e-15\\
35	1.5982641728073e-15\\
36	1.96724298122878e-15\\
37	1.54560232484186e-15\\
38	1.80891213917904e-15\\
39	1.52354103072613e-15\\
40	1.89767438490851e-15\\
41	1.69908197962978e-15\\
42	1.52275062105836e-15\\
43	1.54534617888806e-15\\
44	1.55156964032115e-15\\
45	1.48453754278026e-15\\
46	1.5574445319595e-15\\
47	1.36666471219982e-15\\
48	1.49102559635309e-15\\
49	1.44869999904118e-15\\
50	1.93793377861436e-15\\
51	1.13849994205532e-15\\
52	1.42224784941508e-15\\
53	1.94773080463914e-15\\
54	1.61294672602751e-15\\
55	1.42028641512283e-15\\
56	1.51037274697148e-15\\
57	1.7449144452044e-15\\
58	1.31394042499692e-15\\
59	1.37367862975571e-15\\
60	1.35864249997077e-15\\
61	1.42028506815316e-15\\
62	1.57775441735506e-15\\
63	1.55153850840665e-15\\
64	1.90213733179718e-15\\
65	2.0093445664871e-15\\
66	1.50027360521434e-15\\
67	1.80193012438957e-15\\
68	1.50027347848801e-15\\
69	1.70524480681009e-15\\
70	1.42224334145671e-15\\
71	1.2173868527112e-15\\
72	1.2922362762294e-15\\
73	1.81012784303841e-15\\
74	2.00553163535896e-15\\
75	1.39148893330304e-15\\
76	1.39148892426375e-15\\
};
\addlegendentry{\texttt{sqrtm}, $E_c$}

\end{axis}

\end{tikzpicture}%